%% file: neurips2021.tex
\documentclass[letter-paper]{article}

\usepackage[preprint,nonatbib]{neurips_2021}

\input{headers}

\title{An Uncertainty Principle is a Price of Privacy-Preserving Microdata}

\author{%
  John Abowd\\
  U.S. Census Bureau\\
  and Cornell University\\
  \And Robert Ashmead\\
  U.S. Census Bureau\\
  \And Ryan Cumings-Menon\\
   U.S. Census Bureau\\
   \\
 \And Simson Garfinkel\\
  (formerly) U.S. Census Bureau\\
  U.S. Department of Homeland Security\\
  and George Washington University
  \And Daniel Kifer\\
  U.S. Census Bureau\\
  and Penn State University\\
  \\
  \And Philip Leclerc\\
  U.S. Census Bureau\\
  \\
  \\
  \And William Sexton\\
  (formerly) U.S. Census Bureau\\
  and Tumult Labs\\
  \And Ashley Simpson\\
  Knexus\\
  \\
  \And Christine Task\\
  Knexus\\
  \\
  \And Pavel Zhuravlev\\
  U.S. Census Bureau\\
  \\
}

\begin{document}

\maketitle

\begin{abstract}
\input{abstract}
\end{abstract}

\section{Introduction}\label{sec:intro}
\input{intro}

\section{Preliminaries}\label{sec:prelim}
\input{background}

\section{The Uncertainty Principle}\label{sec:uncertainty}
\input{theory}

\section{Algorithms}\label{sec:algorithms}
\input{algorithms}


\section{Experiments}\label{sec:experiments}
\input{experiments}

\section{Related Work}\label{sec:related}
\input{related}

\section{Conclusions, Future Work, and Broader Impact}\label{sec:conc}
\input{conc}

\input{ack}

\bibliographystyle{plain}

\input{neurips2021.bbl}

\clearpage


\input{appendix}

\end{document}

%% file: headers.tex
\usepackage[utf8]{inputenc} 
\usepackage[T1]{fontenc}    
\usepackage{graphicx}
\usepackage{xcolor}
\usepackage{hyperref}
\usepackage{url}
\usepackage{amsmath}
\usepackage{amsthm}
\usepackage{enumerate}
\usepackage[conf={restate,no link to proof,text proof={}}]{proof-at-the-end}
\usepackage[linesnumbered,ruled,vlined]{algorithm2e}
\usepackage{booktabs}
\usepackage{amsfonts}       
\usepackage{nicefrac}       

\usepackage{microtype}      

\usepackage{multirow}
\usepackage{float}

\newtheorem{theorem}{Theorem}
\newtheorem{definition}{Definition}

\newtheorem{lemma}{Lemma}

\DeclareMathOperator{\range}{range}

\let\vec\mathbf

\DeclareMathOperator{\data}{\mathfrak{D}}

\newcommand{\mystrut}{\rule[-.3\baselineskip]{0pt}{\baselineskip}}

%% file: abstract.tex
Privacy-protected microdata are often the desired output of a differentially private algorithm since  microdata is familiar and convenient for downstream users. However, there is a statistical price for this kind of convenience. We show that an uncertainty principle governs the trade-off between accuracy for a population of interest (``sum query'') vs. accuracy for its component sub-populations (``point queries''). Compared to differentially private query answering systems that are not required to produce microdata, accuracy can degrade by a logarithmic factor. For example, in the case of pure differential privacy, without the microdata requirement, one can provide noisy answers to the sum query and all point queries while guaranteeing that each answer has squared error $O(1/\epsilon^2)$. With the microdata requirement, one must choose between allowing an additional $\log^2(d)$ factor ($d$ is the number of point queries) for some point queries or allowing an extra $O(d^2)$ factor for the sum query. We present lower bounds for pure, approximate, and concentrated differential privacy. We propose mitigation strategies and create a collection of benchmark datasets that can be used for public study of this problem.

%% file: intro.tex
Differential Privacy \cite{dwork06Calibrating} is a mathematical theory of information leakage that allows organizations to publish noisy statistics about their datasets while protecting the confidentiality of user information. 
 Its state-of-the-art guarantees have resulted in adoption by data collectors such as the U.S. Census Bureau ~\cite{ashwin08:map,onthemap,Haney:2017:UCF,abowd18kdd}, Google~\cite{rappor,prochlo}, Apple~\cite{applediffp},
Microsoft~\cite{DingKY17}, Uber~\cite{elasticsensitivity}, and Facebook
\cite{fburl}.

In many cases, downstream users want the output of disclosure avoidance systems in the form of microdata (a set of records about individuals). For example, this has historically been the case for tabulations of Census Bureau data, and is currently a requirement for most 2020 Census of Population and Housing tabulations\cite{issues}. However, an end-user study of demonstration data products released by an early prototype of the Census Bureau's disclosure avoidance system showed significant anomalies in the privacy-protected microdata \cite{cnstatworkshop11}.\footnote{Throughout this paper we use \emph{privacy-protected} and \emph{privacy-preserving} synonymously. The Census Bureau prefers ``privacy-protected,'' whereas the scientific literature has more often used ``privacy-preserving.'' Both terms mean that the confidentiality of individual responses has been protected using differentially private algorithms.} They noted the following: the system first produced differentially private noisy query answers, called \emph{measurements}, and then synthesized privacy-protected microdata so that query answers computed from the privacy-protected microdata matched the noisy measurements as closely as possible (based on some objective function). However, after the privacy-protected microdata were created, they compared  (1) the original measurement query noisy answers and (2) the values of the same queries computed from the privacy-protected microdata. They noted that in some cases, the query error from the privacy-protected microdata was ``much larger'' than the measurement query error \cite{cnstatworkshop11}.

In this paper, we show that such  anomalies are an inherent and unavoidable consequence of privacy-protected microdata (they affect all differentially private algorithms that must output microdata). We further show that the additional errors caused by privacy-protected microdata also satisfy a \emph{new} uncertainty principle that trades off error between accuracy on populations and accuracy on sub-populations. We next explain this principle.

First, our criterion is \emph{per-query} expected squared error. That is, if $Q$ is a collection of queries, $\data$ is the true data, and $\widetilde{\data}$ is the privacy-protected microdata, we are interested in the left side of Equation \ref{eqn:themetrics} (below), where the expectation is taken over the randomness of the algorithm that ingests $\data$ and outputs privacy protected $\widetilde{\data}$. 

\begin{align}
\underbrace{\max_{q\in Q} E_{\widetilde{\data}}[(q(\data)-q(\widetilde{\data}))^2]}_{\text{Our focus: per-query error}} 
 \leq 
\underbrace{E_{\widetilde{\data}}[\max_{q\in Q}(q(\data)-q(\widetilde{\data}))^2]}_{\text{Most other papers: simultaneous/outlier error.}}
\label{eqn:themetrics}.
\end{align} 
This metric measures whether there exist ``bad'' queries that have systematically large errors \emph{on average}. It is \emph{not} to be confused with simultaneous/outlier noise error (right side of Equation \ref{eqn:themetrics}) that is the focus of most theoretical papers on differential privacy, such as \cite{BLR08}. The reason is that simultaneous error cannot distinguish between systematic error in specific queries vs. outliers that result by chance when dealing with many random variables. On the other hand  per query-error can make this distinction because it considers the average behavior of each query separately.

Next, consider a collection of $d$ disjoint\footnote{That is, adding/removing a record into the data can only affect the answer to \underline{\textbf{one}} of the queries.} counting queries  $q_1,\dots, q_d$ and a special query $q_*$ that is equal to their sum ($q_*(\data)=\sum_i q_i(\data)$). We call $q_1,\dots, q_d$ the \emph{point queries} and $q_*$ the \emph{sum} query. Examples include (1) $q_*(\data)=$ ``\# of Black or African Americans in the data living in California'' and $q_i(\data)=$ ``\# of Black or African Americans in the data living \underline{in county $i$} in California'' and (2) $q_*(\data)=$ ``population of a given county'' (which can be  used in federal and state-level funding allocations) and $q_i(\data)=$ ``population in census block $i$ in that county'' (useful for redistricting). Thus, for different use-cases, accuracies at these local and aggregate scales are important.

It is well-known that queries $q_1,\dots, q_d, q_*$ can be answered using $\epsilon$-differential privacy by adding Laplace$(2/\epsilon)$ noise to each query \cite{diffpbook}, thus guaranteeing that each query answer has expected squared error $8/\epsilon^2$. However, in this paper, we show that it is not possible to guarantee this kind of error if one is required to produce differentially private microdata $\widetilde{\data}$ and answer queries using it (i.e., computing $q_1(\widetilde{\data}),\dots, q_d(\widetilde{\data}), q_*(\widetilde{\data})$). Specifically, suppose an $\epsilon$-differentially private microdata-producing algorithm can guarantee that, for all datasets $\data$, $E_{\widetilde{\data}}[(q_*(\data)-q_*(\widetilde{\data}))^2]\leq D^2$ and $\max_i E_{\widetilde{\data}}[(q_i(\data) - q_i(\widetilde{\data}))^2]\leq C^2$ for some constants $C$ and $D$. Then one has to choose:
\begin{itemize}
\item If $D^2 \in O(1/\epsilon^2)$ then $C^2\in\Omega(\frac{1}{\epsilon^2}\log^2(d))$. That is, making the sum query accurate may force us to take a $\log^2(d)$ penalty in the expected squared error some of the point queries, or
\item If $C^2 \in O(1/\epsilon^2)$ then $D^2\in \Omega(\frac{d^2}{\epsilon^2})$. That is, a low per-query error guarantee for point queries may increase expected squared error of the sum query  by a factor of $d^2$.
\end{itemize}

We present such lower bound results for pure differential privacy \cite{dwork06Calibrating}, approximate differential privacy \cite{ourdata}, and concentrated differential privacy \cite{zcdp}, with nearly matching upper bounds.

We note that this uncertainty principle affects some, but not all, possible datasets. That is, there are datasets for which the error penalties do not exist. Thus, the goal in practical privacy-protected microdata generation should be to minimize the occurrence of this uncertainty principle (since eliminating it entirely is impossible). To this end, we propose a  benchmark suite of real and synthetic datasets that can be used by the wider community for further study of this problem. We also propose some algorithms, inspired by our lower and upper bound proofs, for mitigating the effects of this uncertainty principle. Limitations: empirically, these algorithms perform well on the benchmarks but we do not have theoretical proofs of performance.

%% file: background.tex
Let $\data$ denote a dataset, $M$ a differentially private algorithm, and let $\widetilde{\data}$ be a privacy-preserving dataset (e.g., $M(\data)=\widetilde{\data}$). A counting query $q$ is associated with a predicate $\psi$, and the query answer $q(\data)$ is the number of records in $\data$ that satisfy $\psi$. We let $q_1,\dots, q_d$ represent a set of $d$ counting queries whose corresponding predicates $\psi_1,\dots,\psi_d$ are \textbf{disjoint} (no record can satisfy more than one of the predicates). We also let $q_*$ denote their sum: $q_*(\data)=\sum_{i=1}^d q_i(\data)$.

\subsection{Differential Privacy}\label{subsec:dp}

Differential privacy is currently considered the gold standard in privacy protections.
It relies on the concept of neighboring datasets, defined as follows.
\begin{definition}[Neighbors]
Two datasets $\data_1$ and $\data_2$ are neighbors, denoted by $\data_1\sim \data_2$, if $\data_1$ can be obtained from $\data_2$ by adding or removing one record.
\end{definition}
Using this concept of neighbors, differential privacy ensures that adding or removing one record from a dataset has little effect on the probabilistic outcomes of an algorithm: 
\begin{definition}[Differential Privacy \protect\cite{dwork06Calibrating}]
Given privacy parameters $\epsilon>0$ and $\delta\geq 0$, a randomized algorithm M satisfies $(\epsilon,\delta)$-DP if for all pairs of datasets $\data_1, \data_2$ that are neighbors of each other, and for all $S\subseteq\range(M)$, the following equation holds:
\begin{align*}
    P(M(\data_1)\in S) \leq e^{\epsilon}P(M(\data_2)\in S)+\delta,
\end{align*}
where the probability is only over the randomness in $M$ (not the randomness in the data).
When $\delta=0$, we say that $M$ satisfies \emph{pure differential privacy} (also known as $\epsilon$-differential privacy or $\epsilon$-DP) and when $\delta>0$ we say that $M$ satisfies \emph{approximate differential privacy}.
\end{definition}

Another important version of differential privacy, is 
$\rho$-zCDP (concentrated differential privacy):
\begin{definition}[zCDP \protect\cite{zcdp}]
Given a privacy parameter $\rho$, a randomized algorithm $M$ satisfies $\rho$-zCDP if for all pairs of datasets $\data_1, \data_2$ that are neighbors of each other and all numbers $\alpha > 1$,
\begin{align*}
    \mathcal{D}_\alpha(M(\data_1)||M(\data_2)) &\leq \rho\alpha
\end{align*}
where $\mathcal{D}_\alpha(P || Q) \equiv  \frac{1}{\alpha-1}\log\left(E_{x\sim P}\left[\frac{P(x)^{\alpha-1}}{Q(x)^{\alpha-1}}\right]\right)$ is the Renyi divergence of order $\alpha$ between probability distributions $P$ and $Q$. 
\end{definition}
Although zCDP is difficult to interpret, there are useful results that help provide intuition. First, any $M$ that satisfies $\epsilon$-differential privacy also satisfies $\rho$-zCDP with $\rho=\frac{\epsilon^2}{2}$ \cite{zcdp}. In general a $\rho$-zCDP algorithm does not satisfy pure differential privacy but does satisfy $(\epsilon, \delta)$-DP for infinitely many pairs of $\epsilon$ and $\delta$ that lie along a curve (see \cite{dgm} and \cite{alcks20} for conversions between $\rho$-zCDP and $(\epsilon,\delta)$-DP).

\subsection{Algorithm Design with Differential Privacy}
A few basic principles underlie the construction of many algorithms for differential privacy.
The first is sensitivity, which measures the maximum impact that one record can have on a set of queries (regardless of input data):
\begin{definition}[Sensitivity \protect\cite{dwork06Calibrating}]
The $L_p$ global sensitivity of a set $Q$ of queries, denoted by $\Delta_p(Q)$, is defined as $\sup\limits_{\data_1\sim \data_2} \left(\sum_{q\in Q}|q(\data_1)-q(\data_2)|^p\right)^{1/p}$. 
\end{definition}

Global sensitivity can be used with the Laplace and Gaussian distributions to form basic mechanisms. Let Lap$(\alpha)$ represent a draw from the Laplace distribution with density $f(x)=\frac{1}{2\alpha}e^{-|x|/\alpha}$ and $N(0,\sigma^2)$ represent the zero-mean Gaussian distribution with variance $\sigma^2$. Each appearance of $Lap(\alpha)$ or $N(0,\sigma^2)$ represents an independent sample from the corresponding distribution. 

\begin{theorem}[Laplace Mechanism \protect\cite{dwork06Calibrating}]
Given a privacy parameter $\epsilon>0$, a set $Q$ of queries, and an input dataset $\data$, the mechanism $M$ that returns the set of noisy answers $\{q(\data)+Lap(\Delta_1(Q)/\epsilon)\}_{q\in Q}$ satisfies $\epsilon$-differential privacy.
\end{theorem}

\begin{theorem}[Gaussian Mechanism \protect\cite{zcdp}]
Given a privacy parameter $\epsilon>0$, a set $Q$ of queries, and an input dataset $\data$, the mechanism $M$ that returns the set of noisy answers $\{q(\data)+N(0, \Delta_2(Q)^2/(2\rho))\}_{q\in Q}$ satisfies $\rho$-zCDP.
\end{theorem}

All of these privacy definitions are postprocessing invariant \cite{diffpbook}. That is, let $A$ be an arbitrary algorithm. Then $A\circ M$ (i.e., the algorithm that outputs $A(M(\data))$) satisfies $(\epsilon,\delta)$-DP (resp., $\rho$-zCDP) if $M$ satisfies $(\epsilon,\delta)$-DP (resp., $\rho$-zCDP); in other words, the privacy parameters do not degrade.

They also have useful sequential composition properties. Let $M_1,\dots, M_k$ be algorithms that satisfy pure differential privacy with corresponding parameters $\epsilon_1,\dots, \epsilon_k$ (resp., zCDP with corresponding privacy parameters $\rho_1,\dots, \rho_k$), then the algorithm  $M$ that releases all of their outputs (i.e., releases $M_1(\data), \dots, M_k(\data)$) satisfies $\sum_i\epsilon_i-$differential privacy \cite{diffpbook} (resp., $\sum_i\rho_i$-zCDP \cite{zcdp}).

%% file: theory.tex
The setting of $d$ disjoint queries $q_1,\dots, q_d$ and their sum $q_*$ are some of the most important types of query sets. As discussed earlier, population counts in small geographic regions such as census blocks (examples of $q_i$) are important for redistricting while population counts in larger regions such as counties (examples of $q_*$) are used for federal and state funding formulas. Thus any tension between the $q_i$ and $q_*$ can have significant impact on the entire U. S. population. While this is just one example of a query set, almost every table produced in previous censuses is a query set with disjoint queries and their sums \cite{census2010}. Thus this is an important collection of queries to study.

\subsection{Lower Bounds}
We first remove some restrictions on $M$. While its input is a dataset, its output can be a positively
 weighted dataset -- a collection of records in which each record $r$ has a nonnegative weight $w$. A query $q$ with predicate $\psi$ can be evaluated over a weighted dataset by summing the weights of the records that satisfy $\psi$.
 This simplifies our proofs and slightly increases generality,
 since normal microdata is a special case of positively weighted data in which all weights are 1 (hence lower bounds for positively weighted data are also lower bounds for normal microdata). It also emphasizes the fact that these lower bounds arise specifically because negative query answers are disallowed.
 The lower bound  is the following (see supplementary material for proofs).

\begin{theoremEnd}[category=lowerbound]{theorem}\label{thm:lowerbound}
Let $q_1\dots, q_d$ be a set collection of disjoint queries and let $q_*$ be their sum.
Let $M$ be a randomized algorithm  whose input is a dataset and whose output is a positively weighted dataset. Suppose $M$ guarantees that for each query $q_i$ and dataset $\data$, $E[(q_i(\data) - q_i(M(\data)))^2] \leq C^2$ and $E[(q_*(\data) - q_*(M(\data)))^2] \leq D^2$ for some values $C$ and $D$, where the expectation is \textbf{only}  over the randomness in $M$. 
\begin{itemize}
\item If $M$ satisfies $\epsilon$-differential privacy then for any $k>0$,  we have   $e^{2\epsilon (2C + k)} \geq \frac{k(d-1)}{16C + 8D + 4k}$ which implies
\underline{\textbf{(a)}} if $D^2 \leq \lambda/\epsilon^2$ for some constant $\lambda$, then $C^2\in\Omega(\frac{1}{\epsilon^2}\log^2(d))$, and \underline{\textbf{(b)}} if  $C\leq \lambda/\epsilon^2$ then  $D\in\Omega(d^2/\epsilon^2)$.

\item If $M$ satisfies $(\epsilon,\delta)$-DP then for any $k>0$, we have $\left(\frac{\delta}{\epsilon} + \frac{4C + 2D + k}{k(d-1)}\right)e^{4\epsilon C + 2k\epsilon} \geq 1/4$, which implies
\underline{\textbf{(a)}} if $D^2\leq \lambda/\epsilon^2$ for some constant $\lambda$, then $C^2\in\Omega\left(\min( \frac{1}{\epsilon^2}\log^2(d),\;\frac{1}{\epsilon^2}\log^2\frac{\epsilon}{\delta} )\right)$; \underline{\textbf{(b)}} if $C\leq \lambda/\epsilon^2$ then either $\epsilon\in O(\delta)$ or $D^2\in\Omega(d^2/\epsilon^2)$.

\item If $M$ satisfies $\rho$-zCDP, then the tradeoff function between $C$ and $D$ (which is more complex and omitted due to space constraints) implies:
\underline{\textbf{(a)}} if $D^2\leq \lambda/\rho$ for some $\lambda$, then $C^2\in\Omega\left(\log(d)/\rho\right)$, and \underline{\textbf{(b)}} if $C^2\leq \lambda/\rho$, then for any $\gamma\in (0,1)$,  we must have $D^2\in \Omega(d^{2\gamma}/\rho)$.
%
\end{itemize}
\end{theoremEnd}
\begin{proofEnd}
The lower bounds for pure and approximate DP (but not zCDP) can be proved as consequences of the work of Balcer and Vadhan \cite{Balcer_Vadhan_2019}. To make this material more self-contained, we write out a direct proof of the lower bounds by borrowing their proof technique.

For notational convenience, we will let $\vec{x}[i]$ denote $q_i(\data)$ (so $\sum_i \vec{x}[i] = \sum_i q_i(\data)=q_*(\data)$). Similarly, we let $\widetilde{\vec{x}}[i]$ denote $q_i(\widetilde{\data})$ (so $\sum_i \widetilde{\vec{x}}[i] = \sum_i q_i(\widetilde{\data})=q_*(\widetilde{\data})$). Thus the vector $\vec{x}$ represents the true point query answers and $\widetilde{\vec{x}}$ represents the privacy protected point query answers. In particular, $\vec{x}$ is a vector of nonnegative integers and $\widetilde{\vec{x}}$ is a vector of nonnegative real numbers.

All probabilities are taken with respect to only the randomness in $M$.

In this proof, \underline{$\alpha$, $\beta$, and $k$} are constants that we will set later.
For any fixed $j$, by Markov's inequality, 
\begin{align}
P(|\widetilde{\vec{x}}[j] - \vec{x}[j]| \geq \alpha C) \leq \frac{E\left[(\vec{x}[j] - \widetilde{\vec{x}}[j])^2 \right]}{C^2\alpha^2}&\leq \frac{1}{\alpha^2}\label{eq:markovC}\\
P\left(\Big|\sum_{i=1}^{d}\widetilde{\vec{x}}[i] - \sum_{i=1}^{d}\vec{x}[i]\Big| \geq \beta D\right) &\leq  \frac{1}{\beta^2}\label{eq:markovD}
\end{align}

For each positive integer $n$, positive number $k$, and  $i=2,\dots, d$, define the set 
\begin{align*}
G_{i,n,k}&=\left\{\widetilde{\vec{x}} ~:~
\substack{\widetilde{\vec{x}}[i] \in [k, k+2\alpha C],\\ \widetilde{\vec{x}}[1]\in [n-2\alpha C -k, n-k], \\\sum_{j=1}^{d}\widetilde{x}[j] \in [n-\beta D, n+\beta D]}\right\}
\end{align*}
The intuition behind the meaning of $G_{i,n,k}$ is that suppose we had a dataset $\data_i$ with vector $\vec{x}_i$ of point query answers where the $i^\text{th}$ point query  satisfied $\vec{x}_i[i]=k+\alpha C$ and $\vec{x}_i[1]=n-k-\alpha C$ (all other entries are 0) then $G_{i,n,k}$ is the set of all possible outputs $M(\data_i)$ where the $1^\text{st}$ and $i^\text{th}$ entries are within $\alpha C$ of their true value and the sum is within $\beta D$ of its true value.

For any fixed $n$ and $k$, we next examine  how many $G_{i,n,k}$ a vector $\widetilde{\vec{x}}$ can belong to (i.e., an overlap condition). A necessary condition for $\widetilde{\vec{x}}$ to belong to some $G_{i,n,k}$ is that $\widetilde{\vec{x}}[i]\geq k$ and $\widetilde{\vec{x}}[1]\geq n-2\alpha C-k$. This means that after assigning the minimal necessary mass to the $1^\text{st}$ element, there is at most $k+ 2\alpha C  +\beta D$ mass to assign to the other elements (without exceeding the upper limit of $n+\beta D$  on the sum of  all the cells). Since at least $k$ units of this mass must be assigned to the $i^\text{th}$ element in order for $\widetilde{\vec{x}}$ to belong to $G_{i,n,k}$, this means that $\vec{\widetilde{x}}$ can belong to $G_{i,n,k}$ for at most $\frac{2\alpha C+\beta D +k}{k}$ different choices of $i$. 

Now define $ \vec{x}_1,\dots, \vec{x}_{d}$ as follows. $\vec{x}_1[1]=n$ with all other entries being $0$. Next for $i=2,\dots,d$ we set $\vec{x}_i[1]=n-\alpha C - k$ and $\vec{x}_i[i]=\alpha C + k$ and all other entries of $\vec{x}_i$ are 0. For each $i$, Let $\data_i$ be a database whose point query answers are $\vec{x}_i$, which is possible since the point queries are disjoint (and this means that $\data_1$ differs from all of the others by the addition/removal of at least $2(\alpha C + k)$ records).

\textbf{For pure differential privacy}, we have:
\begin{align*}
   1 &\geq P\left(M(\data_1) \in \bigcup_{i=2}^{d}G_{i,n,k}\right)
    \geq \frac{k}{2\alpha C + \beta D+k} \sum_{i=2}^{d} P\left(M(\data_1) \in G_{i,n,k}\right) \text{ by overlap condition}\\
    &\geq e^{-\epsilon 2(\alpha C+k)} \frac{k}{2\alpha C + \beta D+k}\sum_{i=2}^{d} P\left(M(\data_i) \in G_{i,n,k}\right)
    \text{ by group privacy property of $\epsilon$-DP \cite{diffpbook}}\\
    &\geq e^{-\epsilon 2(\alpha C+k)} \frac{k}{2\alpha C + \beta D+k}\sum_{i=2}^{d} \left(1 - \frac{2}{\alpha^2}-\frac{1}{\beta^2}\right)
    \text{ by the Markov inequality and union bound}\\
    &= e^{-\epsilon 2(\alpha C+k)} \frac{k(d-1)}{2\alpha C + \beta D+k} \left(1 - \frac{2}{\alpha^2}-\frac{1}{\beta^2}\right)
\end{align*}

Now we set $\alpha=2$ and $\beta=2$ to get 

\begin{align*}
    e^{2\epsilon (2C + k)} \geq \frac{k(d-1)}{16C + 8D + 4k} 
\end{align*}
If $D$ is allowed to be  $\leq C$, then we set $k=C$ and get
\begin{align*}
    e^{6\epsilon C} &\geq \frac{d-1}{28} &&\Rightarrow& C &\geq \frac{1}{6\epsilon}\log\frac{d-1}{28}
\end{align*}
In general, if $D\in O(C)$ (i.e., $D$ is allowed to be at most some constant times $C$) then similar arguments show $C\in \Omega\left(\frac{1}{\epsilon}\log(d)\right)$.

If $D$ is allowed to be $> C$ then we set $k=1/\epsilon$ and get
\begin{align*}
    e^{4\epsilon C + 2} &\geq \frac{(d-1)}{24\epsilon D + 4} &&\Rightarrow &C &\geq \frac{1}{4\epsilon} \left(\log\left(\frac{(d-1)}{24\epsilon D + 4}\right)- 2\right)
\end{align*}
In general, if $D\in\Omega(C)$ (i.e., $D$ is required to be at least some constant times $C$) then similar arguments show that $C\in\Omega\left(\frac{1}{\epsilon}\log(\frac{d}{\epsilon D})\right)$.

Putting these facts together, we see that if $D\in O(1/\epsilon)$ then $C\in\Omega(\frac{1}{\epsilon}\log(d))$. Meanwhile, if $C\in O(1/\epsilon)$ then we must have $D\in\Omega(d/\epsilon)$.


\textbf{For approximate, differential privacy}, using the group privacy property of approximate differential privacy \cite{Balcer_Vadhan_2019},
\begin{align*}
    1 &\geq P\left(M(\data_1) \in \bigcup_{i=2}^{d}G_{i,n,k}\right)
    \geq \frac{k}{2\alpha C + \beta D+k} \sum_{i=2}^{d} P\left(M(\data_1) \in G_{i,n,k}\right) \\
    &\geq  \frac{k}{2\alpha C + \beta D+k}\sum_{i=2}^{d} \left(e^{-\epsilon 2(\alpha C+k)}P\left(M(\data_i) \in G_{i,n,k}\right) - \delta/\epsilon\right) \text{ by group privacy}\\
    &\geq  \frac{k}{2\alpha C + \beta D+k}\sum_{i=2}^{d} \left(e^{-\epsilon 2(\alpha C+k)}\left(1-\frac{2}{\alpha^2}-\frac{1}{\beta^2}\right) - \delta/\epsilon\right) \text{ Markov inequality, union bound}\\
    &=  \frac{k(d-1)}{2\alpha C + \beta D+k} \left(e^{-\epsilon 2(\alpha C+k)}\left(1-\frac{2}{\alpha^2}-\frac{1}{\beta^2}\right) - \delta/\epsilon\right)\\
    %
\end{align*}
Setting $\alpha=2$ and $\beta=2$ gives 
\begin{align*}
1 \geq \frac{k(d-1)}{4 C + 2 D+k} \left(\frac{1}{4}e^{-\epsilon 2(2 C+k)} - \delta/\epsilon\right)
\text{ and so }
    \left(1+\frac{\delta}{\epsilon}\; \frac{k(d-1)}{4 C + 2 D+k} \right)e^{4\epsilon C + 2k\epsilon} &\geq \frac{1}{4}\;\frac{k(d-1)}{4 C + 2 D+k} 
\end{align*}
and this is the same as 
\begin{align*}
\left(\frac{\delta}{\epsilon} + \frac{4C + 2D + k}{k(d-1)}\right)e^{4\epsilon C + 2k\epsilon} \geq 1/4
\end{align*}
Noting that for any $z$, $1+z\leq 2\max\left(1, z\right)$
and so
\begin{align*}
e^{4\epsilon C + 2k\epsilon} &\geq \frac{1}{8}\min\left(\frac{k(d-1)}{4 C + 2 D+k},\; \frac{\epsilon}{\delta}\right)\\
%
\end{align*}
Proceeding as we did for pure differential privacy,
if $D$ is allowed to be $O(C)$, then $C\in \Omega\left(\min\left( \frac{1}{\epsilon}\log(d),\;\frac{1}{\epsilon}\log\frac{\epsilon}{\delta} \right)\right)$; if $D$ is allowed to be $\Omega(C)$ then $C\in\Omega\left(\min(\frac{1}{\epsilon}\log\frac{d}{\epsilon D},\;\frac{1}{\epsilon}\log\frac{\epsilon}{\delta})\right)$.

Putting this together, if $D\in O(1/\epsilon)$ then $C\in\Omega\left(\min\left( \frac{1}{\epsilon}\log(d),\;\frac{1}{\epsilon}\log\frac{\epsilon}{\delta} \right)\right)$ and if $C\in O(1/\epsilon)$ then either $\epsilon\in O(\delta)$ or $D\in\Omega(d/\epsilon)$.

\textbf{For $\rho$-zCDP}, consider a random variable $X$ that is uniform over $\data_2,\dots,\data_{d}$ (i.e., with probability $1/(d-1)$, $X$ is the dataset $\data_i$). 
%
Note that the $\data_i$ we have been using can be constructed so that  $i\neq j$, $\data_i$ and $\data_j$ differ on the addition/removal of $2aC + 2k$ people. Let $I(\cdot; \cdot)$ denote mutual information and $H(\cdot)$ denote entropy. 
By the group privacy property of zCDP \cite{zcdp} we have two facts relating group privacy to mutual information: (1) $\rho (2\alpha C + 2k)^2  \geq I(M(\vec{\data}_i); M(\vec{\data}_j))$ for all $i$ and $j$ (from Proposition 5.3 proof in \cite{zcdp}) and (2) the corollary that $\rho (2\alpha C + 2k)^2 \geq I(X, M(X))$ (from Proposition 6.1 proof in \cite{zcdp}). Then

%
\begin{align}
\rho (2\alpha C + 2k)^2
&\geq    I(X; M(X))
= H(X) - H(X ~|~M(X))\nonumber\\
&= \log_2(d-1)- H(X ~|~M(X))\label{eq:zcdp_part1}
\end{align}
and now we need  to upper bound $H(X~|~M(X))$.  Define  $G$ to be the event that $M(X)$ is in the  $G_{i,n,k}$ that corresponds to the realized value of $X$ (i.e., the event $X=\data_j \Rightarrow M(X)\in G_{j,n,k}$ for $j=2,\dots, d$). Then we obtain a Fano-like inequality (following the proof structure in \cite{cover}) as follows:
\begin{align}
    H(X~|~M(X))&= H(X~|~M(X)) + H(G~|~X, M(X))\nonumber\\
    &\quad\text{(the last entropy is 0 since $G$ is a deterministic function of $X$ and $M(X)$)}\nonumber\\
    &= H(G, X ~|~M(X))
    \quad\text{ by the chain rule for conditional entropy}\nonumber\\
    &= H(G~|~M(X)) + H(X ~|~ G, M(X))\quad \text{ by chain rule for conditional entropy}\nonumber\\
    &\leq 1 +  H(X ~|~ G, M(X)) \quad \text{ since G is binary, its entropy is $\leq 1$}\nonumber\\
    &=1 + P(G=0)H(X~|~M(X), G=0)  + P(G=1)H(X~|~M(X), G=1)\nonumber\\
    &\leq 1 + P(G=0)\log_2(d-1)+ P(G=1)H(X~|~M(X), G=1)\nonumber\\
    &\quad\text{(since the entropy of $X$ is $\leq \log_2(d-1)$  )}\nonumber\\
    &\leq 1 + P(G=0)\log_2(d-1)+ P(G=1)\log_2\left(\frac{2\alpha C + \beta D + k}{k}\right)\nonumber\\
    &\text{(This follows from $G=1$, by the overlap condition, since then $M(X)$ can }\nonumber\\
    &\text{belong to at most $\frac{2\alpha C + \beta D + k}{k}$ of the $G_{i,n,k}$ so conditioned on knowing $M(X)$}\nonumber\\
    &\text{there are at most $\frac{2\alpha C + \beta D + k}{k}$ possible choices for $X$ }\nonumber\\ 
    &\text{and hence  $\log_2$ of this quantity upper bounds the conditional entropy)}\nonumber\\
    &\leq 1 + P(G=0)\log_2(d-1) + \log_2\left(\frac{2\alpha C + \beta D + k}{k}\right)\nonumber\\
    &\leq 1 + (\frac{2}{\alpha^2} + \frac{1}{\beta^2})\log_2(d-1) + \log_2\left(\frac{2\alpha C + \beta D + k}{k}\right)\label{eq:zcdp_part2}
\end{align}
Where the last inequality follows from the Markov inequality and  union bound on $P(G=0)$.
Now, setting $\alpha=\beta=2$ and combining Equations \ref{eq:zcdp_part1} and \ref{eq:zcdp_part2}, we have:

\begin{align*}
    \rho(4 C + 2k)^2 \geq \frac{1}{4}\log_2(d-1) - \log_2\left(\frac{4 C + 2 D + k}{k}\right) - 1
\end{align*}

\underline{If $D$ is allowed to be $\leq C$\mystrut}, we set $k=C$ to get

\begin{align*}
    \rho(6C)^2 &\geq \frac{1}{4}\log_2(d-1) - \log_2\left(\frac{4 C + 2 D + C}{C}\right) - 1 \geq \frac{1}{4}\log_2(d-1) - \log_2\left(\frac{7C}{C}\right) - 1\\
    &\geq \frac{1}{4}\log_2(d-1)-4
    \Rightarrow C\geq\frac{1}{6\sqrt{\rho}}\sqrt{\frac{\log_2(d-1)}{4} - 4}
\end{align*}

so in general, if $D\in O(C)$ then similar arguments show that $C\in\Omega\left(\sqrt{\frac{1}{\rho}\log(d)}\right)$.

\underline{If $D$ is allowed to be $> C$\mystrut}, then let $\gamma$ be any number strictly between $0$ and $1$. 
Then we set $k=1/\sqrt{\rho}$,  and $\alpha=\beta=\sqrt{\frac{3}{(1-\gamma)}}$. Combining Equations \ref{eq:zcdp_part1} and \ref{eq:zcdp_part2}
\begin{align*}
\lefteqn{ (2\alpha\sqrt{\rho} C + 2)^2}\\ &\geq \left(1-\frac{2}{\alpha^2} - \frac{1}{\beta^2}\right)\log_2(d-1) - 
\log_2\left(2\alpha\sqrt{\rho} C + \beta\sqrt{\rho} D + 1\right) - 1\\
&=\gamma \log_2(d-1) - 
\log_2\left(2\alpha\sqrt{\rho} C + \beta\sqrt{\rho} D + 1\right) - 1\\
&\geq \gamma \log_2(d-1) - 
\log_2\left(3\sqrt{\frac{3}{1-\gamma}} \sqrt{\rho}D + 1\right) - 1
\end{align*}
which implies
\begin{align*}
    C \geq \sqrt{\frac{1-\gamma}{12}}\sqrt{\frac{\gamma\log_2(d-1) - \log_2(3\sqrt{\frac{3}{1-\gamma}}\sqrt{\rho}D + 1) - 1}{\rho}} - \frac{1}{\sqrt{\rho}}\sqrt{\frac{1-\gamma}{3}}
\end{align*}

In general, if $D$ is allowed to be $\Omega(C)$ then similar arguments show that
$C\in \Omega\left(\sqrt{1-\gamma}\sqrt{\frac{\gamma\log_2(d-1) - \log_2(\sqrt{\frac{3}{1-\gamma}}\sqrt{\rho}D + 1) - 1}{\rho}} - \frac{1}{\sqrt{\rho}}\sqrt{\frac{1-\gamma}{3}}\right)$.

Putting all of this together, if $D\in O(1/\sqrt{\rho})$ then $C\in\Omega\left(\sqrt{\log(d)/\rho}\right)$. But in order to get $C=O(1/\sqrt{\rho})$, we must have $D\in \Omega(d^\gamma/\sqrt{\rho})$.

\end{proofEnd}
Balcer and Vadhan \cite{Balcer_Vadhan_2019} recently showed a statistical price of privacy-preserving release of the top-k counts in a histogram.  They proved an analogous $O(\log^2(d/k))$ penalty for point queries under $\epsilon$-DP (and also results for approximate DP). Interestingly, although they did not consider tradeoffs with the sum query (since its value was assumed to be public in their work), the results in our Theorem \ref{thm:lowerbound} (for $\epsilon$-DP and approximate DP, but not zCDP) can be proved using the result of their Theorem 7.2. 

We also note that the tradeoff functions between $C$ and $D$ in Theorem \ref{thm:lowerbound} show a much stronger result than items (a) and (b) in Theorem \ref{thm:lowerbound}. For example, they rule out the possibility that both $C^2$ and $D^2$ can simultaneously be just slightly larger than $O(1/\epsilon^2)$.
To understand and interpret Theorem \ref{thm:lowerbound},
let us compare to the Laplace and Gaussian mechanisms, which can produce negative query answers, hence are not equivalent to producing positively weighted datasets (hence not covered by Theorem \ref{thm:lowerbound}).

It is easy to see that $\Delta_1(q_1,\dots, q_d, q_*)=2$ and $\Delta_2(q_1,\dots, q_d, q_*)=\sqrt{2}$. Hence, an algorithm $M^\prime_\epsilon$ can add independent Lap$(2/\epsilon)$ noise to each query to satisfy $\epsilon$-DP, and an algorithm $M^\prime_\rho$ can add independent $N(0, 1/\rho)$ noise to each query to satisfy $\rho$-zCDP. Thus $M^\prime_\epsilon$ achieves expected squared error of $8/\epsilon^2$ for $q_*$ and each $q_i$ (i.e., $C^2=D^2=8/\epsilon^2$). Meanwhile $M^\prime_\rho$ achieves $1/\rho$ expected squared error ($C^2=D^2=1/\rho$). These expected error guarantees hold for all datasets $\data$.

Theorem \ref{thm:lowerbound} says that privacy-preserving algorithms $M$ that are required to produce positively weighted datasets cannot guarantee the same low error -- there are input datasets $\data$ for which the expected errors can be significantly larger. In the case of $M$ that satisfy $\epsilon$-DP, if we want low error for the sum query (e.g., $D^2=O(1/\epsilon^2)$, matching the Laplace mechanism), on some datasets we may need to pay a $\log^2(d)$ penalty for some point queries (i.e., there will be specific point queries with consistently large error). On the other hand, if we want low error for the point queries (e.g., $C^2=O(1/\epsilon^2)$) then on some datasets we will pay a $d^2$ penalty on the sum query.

In the case of $\rho$-zCDP, the penalties are smaller. If we want to match the error of the Gaussian mechanism on the sum query, we may need to pay a penalty of $\log(d)$ on point queries; if we want $O(1/\rho)$ expected squared error on each point query, we may need to pay a penalty of nearly $d^2$ on $q_*$.

For approximate DP, the weakest privacy definition here, the degradation factor can be roughly $\log^2(\epsilon/\delta)$ no matter how large $d$ is.

\paragraph{Remark 1.} The lower bounds in Theorem \ref{thm:lowerbound} imply that if  privacy-preserving microdata is generated by  obtaining noisy measurement query answers (e.g., with the Laplace or Gaussian mechanisms)  and then postprocessing the noisy answers (e.g., \cite{lioptim,mwem}), some of the measurement queries computed directly from the privacy-preserving microdata will have errors that are larger than their original noisy answers.

\paragraph{Remark 2.} All is not lost, however, as the proofs are based on packing arguments that show that these errors are unavoidable for some difficult datasets (but not all datasets are difficult). An example of a difficult dataset $\data^*$ under pure differential privacy is one for which exactly one of the query answers $q_1(\data^*),\dots, q_d(\data^*)$ equals $\log(d)/\epsilon$ while the other $d-1$ queries equal 0 (clearly, $q_*(\data^*)=\log(d)/\epsilon$). As mentioned earlier, the Laplace mechanism \cite{diffpbook}, which does not produce microdata, can achieve $8/\epsilon^2$ per query error 
although many of the noisy query answers will be negative. However, the proof of Theorem \ref{thm:lowerbound} implies that no algorithm that produces privacy-protected microdata (and hence nonnegative query answers) can do as well on such a dataset. In fact, for this specific difficult dataset $\data^*$, the large error described by Theorem \ref{thm:lowerbound} will either occur for $q_*$ or for that $q_i$ whose  answer on $\data^*$ is $\log(d)/\epsilon$. On the other hand, an easy dataset is one for which $q_1(\data),\dots, q_d(\data)$ are all large, since almost no effort is needed in ensuring that the privacy-protected query answers are nonnegative.

\subsection{Upper Bounds}
These lower bounds are nearly tight, as shown by the upper bounds in Theorem \ref{thm:upperboundA}. The proofs construct postprocessing algorithms that first obtain noisy answers $a_1,\dots, a_d, a_*$ to the queries $q_1,\dots, q_d, q_*$. A postprocessing step converts the $a_i$ and $a_*$ into consistent noisy answers $a_1^\prime, \dots, a_d^\prime, a_*^\prime$ (i.e., they are nonnegative and $\sum_i a^\prime_i=a^\prime_*$). Weighted datasets are constructed from the latter quantities. To get weighted datasets with higher accuracy on point queries, the postprocessing ignores $a_*$ and sets $a^\prime_i=\max\{0, a_i\}$. To obtain synthetic data with higher accuracy on the sum query, $a_*^\prime$ is set to $a_*$ and the $a_i^\prime$ are obtained by minimizing squared distance to the $a_i$ subject to the $a^\prime_i$ being nonnegative and adding up to $a_*$. The full proofs are in the supplementary material.


\begin{theoremEnd}[category=upperbound]{theorem}[Upper bound for pure DP and zCDP]\label{thm:upperboundA}
Let $q_1,\dots, q_d$ be a set of disjoint queries and let $q_*$ be their sum. Given privacy parameters $\epsilon>0$ and $\rho>0$, there exist algorithms $M_\epsilon, M_\rho, M^\prime_\epsilon, M^\prime_\rho$, $M^\prime_{\epsilon,\delta}$ that output a positively weighted dataset and have the following properties:
\begin{enumerate}
\item\label{item:dppoint} $M_\epsilon$ satisfies $\epsilon$-DP, and for all $\data$ and $i$, $E\left[(q_i(M_\epsilon(\data))-q_i(\data))^2\right] \leq 2/\epsilon^2$ and $E\left[(q_*(M_\epsilon(\data))-q_*(\data))^2\right] \leq 2d^2/\epsilon^2$.
\item\label{item:zcdppoint}
$M_\rho$ satisfies $\rho$-zCDP, and for all $\data$ and $i$, $E\left[(q_i(M_\rho(\data))-q_i(\data))^2\right] \leq 1/(2\rho)$ and $E\left[(q_*(M_\rho(\data))-q_*(\data))^2\right] \leq d^2/(2\rho)$.

\item\label{item:dpsum} $M^\prime_\epsilon$ satisfies $\epsilon$-DP, and for all $\data$ and $i$, $E\left[(q_i(M^\prime_\epsilon(\data))-q_i(\data))^2\right]\in O(\log^2(d)/\epsilon^2)$ and $E\left[(q_*(M^\prime_\epsilon(\data))-q_*(\data))^2\right] \in O(1/\epsilon^2)$

\item\label{item:zcdpsum} $M^\prime_\rho$ satisfies $\rho$-zCDP, and for all $\data$ and $i$, $E\left[(q_i(M^\prime_\rho(\data))-q_i(\data))^2\right]\in O(\log(d)/\rho)$ and $E\left[(q_*(M^\prime_\rho(\data))-q_*(\data))^2\right] \in O(1/\rho)$

\item\label{item:appsum} $M^\prime_{\epsilon,\delta}$ satisfies $(\epsilon,\delta)$-DP and for all $\data$ and $i$, $E\left[(q_i(M^\prime_{\epsilon,\delta}(\data))-q_i(\data))^2\right]\in O(\log^2(1/\delta)/\epsilon^2 + 1)$ and $E\left[(q_*(M^\prime_{\epsilon,\delta}(\data))-q_*(\data))^2\right] \in O(1/\epsilon^2)$. Also note $M_\epsilon$ and $M_\epsilon^\prime$ satisfy $\epsilon,\delta$-DP.

\end{enumerate}
\end{theoremEnd}
\begin{proofEnd}
The double-sided geometric mechanism $DGeo(\epsilon)$ is a discrete version of the Laplace distribution, supported over integers, with probability mass function $p(k)= \frac{1-e^{-\epsilon}}{1+e^{-\epsilon}}e^{-\epsilon|k|}$ \cite{GRS09}. It has several useful properties: (a) its mean is 0, (b) its variance is $2\frac{e^{-\epsilon}}{(1-e^{-\epsilon})^2}\leq 2/\epsilon^2$, (c) given an integer-valued query $q$, adding $DGeo(\epsilon/\Delta_1(q))$ to its answer satisfies $\epsilon$-differential privacy.

Similarly, the discrete Gaussian $DGauss(0, 1/(2\rho))$ is a discrete version of the Gaussian distribution \cite{dgm} with several useful properties: (a) its mean is 0, (b) its variance is less than that of $N(0, 1/(2\rho))$, (c) given an integer-valued query $q$, adding $DGauss(0, \Delta_2(q)^2/(2\rho))$ to its answer satisfies $\rho$-zcdp.

\textbf{To prove Item \ref{item:dppoint}},
let $r_1,\dots, r_d$ be records satisfying the predicates for point queries $q_1,\dots, q_d$, respectively. Let $M_\epsilon$ be the algorithm that first computes nonnegative noisy query answers $a_i = \max\{0, q_i(\data) + DGeo(1/\epsilon)\}$ for $i=1,\dots, d$ and then outputs the synthetic dataset $\widetilde{\data}$ that has $a_i$ copies of record $r_i$ for each $i$. 
Note that $M_\epsilon$ does not obtain a noisy answer to $q_*$, and so it satisfies $\epsilon$-differential privacy since $\Delta_1(q_1,\dots, q_d)=1$
Since $q_i(\data)\geq 0$ for all $i$, we have:
\begin{align*}
E\left[(q_i(\data)-q_i(M_\epsilon(\data)))^2\right] &= E\left[(q_i(\data) - \max\{0, q_i(\data)+DGeo(\epsilon)\})^2\right]\\
&\leq E\left[(q_i(\data) -  (q_i(\data)+DGeo(\epsilon)))^2 \right]\leq 2/\epsilon^2\\
\end{align*}
Furthermore 
\begin{align*}
\lefteqn{E\left[(q_*(\data) - q_*(M_\epsilon(\data)))^2\right] = E\left[(\sum_i q_i(\data)-\sum_i q_i(M_\epsilon(\data)))^2\right]}\\
&=  \sum_i E\left[(q_i(\data) - \max\{0, q_i(\data)+DGeo(\epsilon)\})^2\right]\\
&\phantom{=} + 2\sum_{i,j: i < j} E\left[\Big(q_i(\data) - \max\{0, q_i(\data)+DGeo(\epsilon)\}\Big)\right]E\left[\Big(q_j(\data) - \max\{0, q_j(\data)+DGeo(\epsilon)\}\Big)\right]\\
&\leq d \frac{2}{\epsilon^2} + d(d-1)\frac{2}{\epsilon^2} = d^2\frac{2}{\epsilon^2}
\end{align*}

\textbf{To prove Item \ref{item:zcdppoint}}, we use the same proofs as before, except that $M_\rho$ synthesizes $\widetilde{\data}$ using the noisy answers $a_i=q_i(\data) + \max\{0, DGauss(0, 1/(2\rho))\}$. Following essentially the same calculations, we see that the expected squared error of each point query $q_i$ is at most $1/(2\rho)$ and for the sum query $q_*$ it is at most $d^2/(2\rho)$.

\textbf{To prove Item \ref{item:dpsum}}, let $r_1,\dots, r_d$ be records satisfying the predicates for point queries $q_1,\dots, q_d$, respectively. Let $M^\prime_\epsilon$ be the algorithm that does the following. First, it obtains noisy answers for each query: $a_i = q_i(\data) + Lap(2/\epsilon)$ for $i=1,\dots, d$ and $a_*= q_*(\data) + Lap(2/\epsilon)$. (Since $\Delta_1(q_1,\dots, q_d, q_*)=2$, this clear satisfies $\epsilon$-differential privacy). Next, $M$ solves the following optimization problem:

\begin{align*}
\arg\min_{x_1,\dots, x_d} &\frac{1}{2}\sum_{i=1}^d (x_i-a_i)^2\\
 \text{s.t. } &\sum_{i=1}^d x_i = \max\{0, a_*\}\\
 \phantom{\text{s.t. }} &x_i\geq 0, \text{ for $i=1,\dots, d$}
\end{align*}
and creates a privacy protected microdata $\widetilde{\data}$ that consists of the records $r_1,\dots, r_d$ with respective weights $x_1,\dots, x_d$.

Since the sum query is nonnegative and the problem is constrained so that $\sum_i x_i$ is equal to $\max\{0, a_*\}$, clearly $E\left[(q_*(M^\prime_\epsilon(\data)) - q_*(\data))^2\right]\leq 2/\epsilon^2$.

Now let us derive an upper bound on $E\left[(q_i(M^\prime_\epsilon(\data)) - q_i(\data))^2\right]$ for a point query $q_i$.

For each $i$, let $z_i=a_i-q_i(\data)$ and $z_*= a_*-q_*(\data)$ be the actual noises that are added (they are all i.i.d. Laplace$(2/\epsilon)$.

We know from Lemma \ref{lem:cnnls} that the solution $x_i$ have the form $\max\{a_i-\gamma, 0\}$ (which is $\max\{0, q_i(\data)+z_i-\gamma\}$) for some  $\gamma$ such that $\sum_i \max\{a_i-\gamma, 0\}=\max\{0, a*\}$ and note that the left hand side is monotonic in $\gamma$. 

We first find a suitable upper and lower bound on $\gamma$. Define $L=-|z_*| +  \min_i z_i$ and $U=|z_*|+ \max_i z_i$. Then we have:
\begin{align*}
\sum_i \max\{0, a_i - U\} &= \sum_i \max\{0, q_i(D)+z_i-U\}\leq \sum_i\max\{0, q_i(\data) - |z_*|\}\\
&\leq \max\{0, \left(\sum_i q_i(\data)\right) - |z_*|\}\\
&\text{since the $q_i(\data)$ are nonnegative}\\
&= \max\{0, q_*(\data) - |z_*|\}\\
&\leq \max\{0, a_*\}
\end{align*}
and so $\gamma\leq U$.

Next, 
\begin{align*}
\sum_i \max\{0, a_i - L\} &= \sum_i \max\{0, q_i(D)+z_i-L\}\geq \sum_i\max\{0, q_i(\data) + |z_*|\}\\
&= \sum_i\left( q_i(\data) + |z_*|\right)\\
&\text{since the $q_i(\data)$ are nonnegative}\\
&\geq \left(\sum_i q_i(\data)\right) + |z_*|\\
&=q_*(\data) + |z_*|\geq \max\{0, a_*\}
\end{align*}
and so $\gamma \leq L$.

We next find a bound on $E\left[(q_i(M^\prime_\epsilon(\data)) - q_i(\data))^2\right]$ in terms of $\gamma$.
\begin{align}
E\left[(q_i(M^\prime_\epsilon(\data)) - q_i(\data))^2\right] &= E\left[(\max\{0, q_i(\data)+z_i-\gamma\} - q_i(\data))^2\right]\nonumber\\
&\text{ note the random variable here are $z_i$ and $\gamma$}\nonumber\\
&\leq E\left[(( q_i(\data)+z_i-\gamma) - q_i(\data))^2\right]\nonumber\\
&\text{ since $q_i(\data)$ is nonnegative and removing the max moves the}\nonumber\\
&\text{ left part further away from $q_i(\data)$}\nonumber\\
&=E\left[(z_i-\gamma)^2 \right]\leq E\left[\left(|z_i| + \max\{|L|, |U|\}\right)^2\right]\nonumber\\
&\leq E\left[\left(|z_i| + |z|_* + \max_j |z_j|\right)^2\right]\nonumber\\
&\text{since the noises $z_j$ are symmetric around 0}\nonumber\\
&\leq E\left[\left(|z|_* + 2\max_j |z_j|\right)^2\right]\nonumber\\
&=E\left[z_*^2\right] + 4E\left[|z_*|\right]E\left[\max_j |z_j|\right] + 4E\left[(\max_j |z_j|)^2\right]\label{eq:commonproof}\\
&\in O(\frac{1}{\epsilon^2}\log^2(d)) \quad\text{by Lemma \ref{lem:distributions} for Laplace noise}\nonumber
\end{align}

\textbf{To prove Item \ref{item:zcdpsum}} we follow the same steps as before, but using $N(0, 1/(\rho))$ noise instead of Lap$(2/\epsilon)$ (noting that $\Delta_2(q_1,\dots, q_d, q_*)=\sqrt{2}$) and again see that the variance of the sum query is at most $1/\rho$, while for the point queries, the only thing that changes are the calculations after Equation \ref{eq:commonproof}, where we use the Lemma \ref{lem:distributions} results for Gaussian noise, to conclude that
$E[(q_i(\data) - q_i(M(\data)))^2]\in O(\frac{1}{\rho}\log(d))$ for each $i$.

\textbf{To prove Item \ref{item:appsum}}, we again follow the same steps as before but with a different noise distribution. Recall that the double geometric distribution $DGeo(\epsilon)$ is supported over the integers. If $z\sim DGeo(\epsilon)$ then $P(z=k) = \frac{1-e^{-\epsilon}}{1+e^{-\epsilon}}e^{-\epsilon |k|}$. Furthermore, if $k\geq 0$, $P(z \geq k) = P(z \leq -k) =  \frac{1}{1+e^{-\epsilon}}e^{-\epsilon k}$.

For any integer $B> 0$, the truncated double geometric distribution $TDGeo(\epsilon, B)$ is obtained by clipping a DGeo$(\epsilon)$ at $B$ and $-B$. Specifically, if $z'\sim TDGeo(\epsilon, B)$ then
\begin{align*}
P(z^\prime = k) &=
\begin{cases}
\frac{1}{1+e^{-\epsilon}}e^{-\epsilon B} & \text{ if }k=B\\
\frac{1-e^{-\epsilon}}{1+e^{-\epsilon}}e^{-\epsilon |k|} & \text{ for $k=-B+1, \dots, B-1$}\\
\frac{1}{1+e^{-\epsilon}}e^{-\epsilon B} & \text{ if }k=-B\\
\end{cases}
\end{align*}

So, we follow the same approach as in the proof of Item \ref{item:dpsum} but we use $TDGeo(\epsilon/2, B)$ noise to answer each query (detail queries and sum query). We first determine the value of $B$ needed to satisfy $(\epsilon/2,\delta/2)$-DP.

First note that for any integer $v$, and integer $k\in [v-B+1, v+B-2]$
$$e^{-\epsilon/2}\leq \frac{P(v+TDGeo(\epsilon/2, B)=k)}{P(v-1+TDGeo(\epsilon/2, B)=k)} \leq e^{\epsilon/2}$$
(the significance of these points are that they are not in the boundary of $v+TDGeo$ or $v-1+TDGeo$).

Meanwhile $P(v+TDGeo(\epsilon/2)\in\{v-B, v+B-1, v+B\}) = P(DGeo(\epsilon/2) \geq B-1) + P(DGeo(\epsilon/2) \leq -B)= \frac{1}{1+e^{-\epsilon/2}}e^{-\epsilon B/2} + \frac{1}{1+e^{-\epsilon/2}}e^{-\epsilon (B-1)/2}\leq 2e^{-\epsilon (B-1)/2}$. These are the boundary points where the probability ratios may be large.

Setting this equal to $\delta/2$ (and then performing similar calculations when the $v-1$ term is in the numerator), we see that adding $TDGeo(\epsilon/2, B)$ noise satisfies $\epsilon,\delta$-DP if $B\geq \frac{2}{\epsilon}\log(4/\delta)+1$.

Thus using a naive composition result of approximate differential privacy \cite{diffpbook}, we can add TGeo$(\epsilon/2, B)$ noise to each point query and the sum query to satisfy $(\epsilon,\delta)$-DP.

Using the same postprocessing as in the proof of Item \ref{item:dpsum}, we see that the expected squared error of the sum query (when computed from the postprocessed privacy-protected data) is at most variance$(TDGeo(\epsilon/2, B))\leq$ variance$(DGeo(\epsilon/2))\leq 8/\epsilon^2$.

For the point queries, the only thing that changes are the  calculations after Equation \ref{eq:commonproof}. Since the absolute value of the noises is bounded by $B$, we get that the expected squared error of the point queries is $\in O(B^2)=O(\frac{1}{\epsilon^2}\log^2(1/\delta)+1)$.

\end{proofEnd}

Note that Theorem \ref{thm:upperboundA} matches the lower bounds in Theorem \ref{thm:lowerbound} except for a slight difference for zCDP, where Item \ref{item:zcdppoint} of Theorem \ref{thm:upperboundA} has a $d^2$ while the lower bound in Theorem \ref{thm:lowerbound} has in its place a $d^{2\gamma}$ for any $\gamma$ arbitrarily close to 1.

%% file: algorithms.tex
For tabular data, typically end-users are interested in multiple marginals of the data. Examples include the gender by age marginals at the national, state, and county levels (for constructing age pyramids); the marginal on race at the national, state, county, tract, and block levels both for demographic research and for enforcement of voting rights; total populations in each state, county, etc. (for various funding formulas). Thus these query sets have many different point query/sum query collections embedded in them. Examples include: female population in a county (sum query) and number of females of each age in the county (point queries); or total Asian population (sum query) and Asian population in each county (point queries).
Thus algorithms designed to minimize the appearance of the uncertainty principle should not be designed for a \emph{single} collection of sum/point queries; instead, they should support \emph{many} counting queries.

To describe algorithms, it is helpful to view the dataset $\data$ as a vector $\vec{x}$, where each element $i$ corresponds to a possible record $r_i$. Then $\vec{x}[i]$ is the number of times $r_i$ appears in $\data$. The goal is to produce a privacy-protected version $\widetilde{\vec{x}}$ whose entries are nonnegative real numbers, which can be converted to a positively weighted dataset $\widetilde{\data}$ ($\widetilde{\vec{x}}[i]$ is the weight of record $r_i$ in $\widetilde{\data}$). In this setting, a counting query $q$ is just a vector of 1s and 0s with the same dimensionality as $\vec{x}$, and the query answer is computed as the dot product $q\cdot\vec{x}$.

The algorithms we present here (2 baselines and 2 proposed algorithms) are all based on the idea of first computing noisy query answers and then postprocessing them to obtain $\widetilde{\vec{x}}$. This setup allows an organization to release both $\widetilde{\vec{x}}$ and the noisy answers (for more statistically-oriented end-users). 
Thus, given a set $Q$ of counting queries, for each $q\in Q$, the data collector computes a noisy answer $a_q$ by adding noise with distribution $F_q$ to the true answer and then must postprocess them to create microdata.\footnote{Although a data collector could add noise to a different set of queries and use them to infer the answers to $q\in Q$ \cite{YYZH16,MatrixMech,XiaoDWZK21}, it is the subsequent postprocessing step that would be more important in mitigating the uncertainty principle.} We assume the data collector chooses the noise distributions to achieve their desired privacy definition (e.g., $\epsilon$-DP, $\rho$-zCDP).

\textbf{Baseline: NNLS Postprocessing.} The first baseline we consider is the commonly used nonnegative least squares (NNLS), in which $\widetilde{\vec{x}}$ is produced as the solution to the following optimization problem:
\begin{align*}
\widetilde{\vec{x}}\gets \arg\min_{\widetilde{\vec{x}}} \sum_{q\in Q} \frac{(a_q - q\cdot\widetilde{\vec{x}})^2}{variance(F_q)} \text{ s.t. }\widetilde{\vec{x}}[i]\geq 0 \text{ for all $i$} 
\end{align*}

\textbf{Baseline: Max Fitting Postprocessing.} The next baseline is an adaptation of a bilevel optimization approach \cite{Fioretto2020BilevelOF} that was originally used for optimization problems whose parameters are sensitive. The idea here is to find the positively weighted datasets whose query answers minimize the $L_\infty$ distance to the noisy query answers, breaking ties using least squares error:
\begin{align*}
dist\gets &\min_{\widetilde{\vec{x}}} \max_{q\in Q} \frac{|a_q - q\cdot\widetilde{\vec{x}}|}{std(F_q)} \text{ s.t. }\widetilde{\vec{x}}[i]\geq 0 \text{ for all $i$}\\ 
\widetilde{\vec{x}}\gets &\arg\min_{\widetilde{\vec{x}}}  \sum_{q\in Q} \frac{(a_q - q\cdot\widetilde{\vec{x}})^2}{variance(F_q)} \text{ s.t. }\max_{q\in Q} \frac{|a_q - q\cdot\widetilde{\vec{x}}|}{std(F_q)} \leq dist \text{ and }\widetilde{\vec{x}}[i]\geq 0 \text{ for all $i$} 
\end{align*}

\textbf{Sequential Fitting Postprocessing.} Since it is provably not always possible to output microdata that fits the noisy answers well, we propose an approach that prioritize queries. Thus the query set $Q$ is partitioned by the user into query sets $Q_1,\dots, Q_k$. We use the above NNLS approach to fit a vector $\widetilde{\vec{x}}_1$  to the noisy answers of queries in $Q_1$ (highest priority). We then fit  $\widetilde{\vec{x}}_2$ to the noisy answers for queries in $Q_2$ (next highest priority) subject to the constraints that $\widetilde{\vec{x}}_2$ matches $\widetilde{\vec{x}}_1$ on queries in $Q_1$. Then we fit $\widetilde{\vec{x}}_3$ using noisy answers to queries in $Q_3$ while forcing $\widetilde{\vec{x}}_3$ to match $\widetilde{\vec{x}}_2$ on queries in $Q_1$ and $Q_2$, and so on and return the final $\widetilde{\vec{x}}_k$ at the end. The pseudocode is shown in Algorithm \ref{alg:seq}. This algorithm is the one that matches the upper bounds in Theorem \ref{thm:upperboundA} (referred to as $M^\prime_\epsilon$ when the noisy answers $a_q$ use Laplace noise, and $M^\prime_\rho$ for Gaussian noise).

\begin{algorithm}
\caption{Sequential Fitting (Postprocessing)}\label{alg:seq}
\DontPrintSemicolon
\textbf{Input:} Query set $Q$, noisy answers $a_q$ for $q\in Q$ and noise distributions $F_q$ for $q\in Q$.\;
\textbf{Input:} $Q_1,\dots, Q_k$: partition of $Q$ based on query priority.\;
$\widetilde{\vec{x}}_1 \gets \arg\min_{\widetilde{\vec{x}}} \sum_{q\in Q_1} \frac{(a_q - q\cdot\widetilde{\vec{x}})^2}{variance(F_q)} \text{ s.t. }\widetilde{\vec{x}}[i]\geq 0 \text{ for all $i$}$\;
Fit $\gets Q_1$\;
\For{$\ell=2,\dots, k$}{
   $\widetilde{\vec{x}}_\ell \gets \arg\min_{\widetilde{\vec{x}}} \sum_{q\in Q_\ell} \frac{(a_q - q\cdot\widetilde{\vec{x}})^2}{variance(F_q)} \text{ s.t. }\widetilde{\vec{x}}[i]\geq 0 \text{ for all $i$ and } q\cdot\widetilde{\vec{x}} = q\cdot\widetilde{\vec{x}}_{\ell-1}$ for all $q\in $ Fit\;
   Fit $\gets $ Fit $\cup Q_\ell$\;
}
\textbf{Return:} $\widetilde{\vec{x}}_k$
\end{algorithm}

\textbf{Remark.} The constrained optimizations in max fitting and sequential fitting are difficult for quadratic program optimizers, often resulting in numerical errors, slow convergence, and infeasibility errors (due to occasional insufficient solution quality in earlier stages of the multistage optimization). They require significant engineering effort, tuning of slack parameters (slightly relaxing equality and inequality constraints) and optimizer-specific parameters. So, an ideal solution would also avoid constraints other than nonnegativity for point queries. This is a rationale for our next  method.

\textbf{ReWeighted Fitting Postprocessing.} This method (shown in Algorithm \ref{alg:weighted}) avoids constraints as much as possible in an eventual NNLS solve (Line \ref{line:reweightnnls}) but is limited to query sets of the form $Q=\bigcup_{i=1}^k Q_i$, where the queries inside each $Q_i$ are disjoint and have the same noise distribution. One example is when $Q$ is a collection of marginal queries (e.g., $Q_1=$ marginal on age, $Q_2=$ marginal on age by race, $Q_3=$ marginal on gender by race), which are arguably the most important types of queries. Within each $Q_i$, the algorithm tries to find a cutoff value so that queries with noisy answers above it are likely to have true value that is non-zero (Lines \ref{line:cutA}-\ref{line:cutB}). The idea is that if $n_\dagger$ is the number of queries below the threshold, and if they truly had value 0, then their largest noisy value (i.e., the max of $n_\dagger$ 0-mean Laplace or Gaussian random variables) should not be near the cutoff with high probability (controlled by the confidence parameter $\gamma$). The ``low'' queries are the ones with noisy answers below the cutoff. The algorithm uses the \underline{existing} noisy answers to estimate the sum of these ``low'' queries (Lines \ref{line:sumA}-\ref{line:sumB}) and adds that ``low query sum'' (Line \ref{line:addsum}) to the nonnegative least squares optimization while downweighting the individual low queries (Line \ref{line:qsel}, the downweight depends on the extreme value distribution of the max of $n_\dagger$ $0$-mean Laplace or Gaussian random variables, Line \ref{line:downweight}).
To avoid double counting, both places where a ``low'' query is used (individually and as part of a sum) have their weights cut in half. Note the algorithm only uses existing noisy answers and has no access to the true data.


\begin{algorithm}
\caption{ReWeighted Fitting (Postprocessing)}\label{alg:weighted}
\DontPrintSemicolon
\textbf{Input:} Query set $Q=\bigcup_{i=1}^k Q_i$; Within a $Q_i$, the queries are disjoint. $F_i$ is the noise distribution of each query in $Q_i$. Given noisy answers $a_q$ for $q\in Q$ that satisfy the chosen privacy definition.\;
\textbf{Input:} Confidence parameter $\gamma$ close to 1 (e.g., $0.99$, the setting used in experiments)\;
$S\gets\emptyset$
\For{$i=1,\dots, k$}{
  $a_{(1)}, a_{(2)}, \dots $ are the given noisy answers (to queries in $Q_i$) arranged in sorted order\;
  $j^* \gets$ smallest $j$ s.t. $P$(\text{max($j$ fresh random variable with distribution$F_i)\geq a_{(j)}$})$\leq 1-\gamma$\;\label{line:cutA}
  $cutoff$ $\gets a_{(j^*)}$.\;\label{line:cutB}
  $downweight$ $\gets$ median of distribution of  max of $j$ random variables sampled from $F_i$\;\label{line:downweight}
  For each query $q\in Q_i$ whose noisy answer $a_q$ is $\geq$ $cutoff$, add $(q, a_q, 1/var(F_i))$ to $S$.\;
   For each query $q\in Q_i$ whose  $a_q$ is $<$ $cutoff$, add $(q, a_q, \frac{1}{2*var(F_i)*downweight^2})$ to $S$.\label{line:qsel}\;
   $n^\dagger_i \gets $ number of queries selected in Line 9 (i.e., their noisy answers were $<cutoff$)\;
   $q_\dagger\gets $ sum of queries selected in Line
   \ref{line:qsel}\;  \label{line:sumA}
   $a_\dagger \gets$ sum of their existing noisy answers\;\label{line:sumB}
  Add $(q_\dagger, a_\dagger, \frac{1}{2*n^\dagger_i var(F_i)})$ to $S$\;\label{line:addsum}
}
$\widetilde{\vec{x}}\gets \arg\min_{\widetilde{\vec{x}}} \sum_{(q^\prime,a^\prime,w^\prime)\in S} w^\prime (q^\prime(\widetilde{\vec{x}}) - a^\prime)^2$ s.t., $\widetilde{\vec{x}}[i]\geq 0$ for all $i$.\;\label{line:reweightnnls}
\textbf{Return:} $\widetilde{\vec{x}}$
\end{algorithm}

%% file: experiments.tex
To make our code fully open source, we wrote it in Julia \cite{bezanson2017julia} and after trying several open-source optimizers, we settled on COSMO \cite{cosmo}.
We created a collection of benchmark datasets that were small enough to permit running the postprocessing algorithms thousands of times on each dataset (to estimate expected errors) but large enough to demonstrate the uncertainty principle. The full benchmark of 15 real datasets and 16 synthetic datasets is described in the supplementary material.\footnote{See \url{https://github.com/uscensusbureau/CostOfMicrodataNeurIPS2021} for the code and data.} Here we present results for an interesting subset. The only synthetic dataset discussed here, called Level00-2d, is a $10\times 10$ histogram where one element is large (i.e., 10,000) and the others are 0. The other 15 datasets we discuss here were taken from the 2016 ACS Public-Use Microdata Sample \cite{2016pums}. Each represents a $9\times 24$ ``race by Hispanic origin'' histogram from Public-Use Microdata Areas that were considered outliers in their states in terms of racial composition.

For these datasets, we applied the Laplace mechanism with $\epsilon=0.5$ to answer the sum query, both 1-way marginal queries, and identity queries (for each cell, how many people are in it). This is also the priority order used by Sequential Fitting. Error results for the marginals, other privacy parameters and zCDP results can be found in the supplementary material. We ran the Laplace mechanism using  different postprocessing strategies (described in Section \ref{sec:algorithms}) 1,000 times for each dataset to estimate expected squared error of each query. We added an ordinary least squares (OLS) optimization for comparison purposes (OLS is NNLS without nonnegativity constraints). OLS is free from the uncertainty principle because it does not produce positively weighted microdata. Thus, to minimize the effect of the uncertainty principle, the other postprocessing methods should try to achieve errors that are not much worse than OLS. We note that the multi-stage optimization in Max and Sequential fitting are generally very difficult for optimization software, so we only kept those runs in which the optimizer succeeded (thus results for Max and Sequential Fitting are slightly optimistically biased). 

\begin{table}
\begin{tabular}{|c|l|rrrrr|}
\multicolumn{1}{c}{\textbf{Dataset Nickname}}
&\multicolumn{1}{c}{\textbf{Dataset}}
 & \multicolumn{1}{c}{\textbf{OLS}}
 & \multicolumn{1}{c}{\textbf{NNLS}}
 & \multicolumn{1}{c}{\textbf{MaxFit}}
 & \multicolumn{1}{c}{\textbf{Seq}}
 & \multicolumn{1}{c}{\textbf{ReWeight}}
\\\hline 
$\data$01 &
Level00-2d
 & 101.3 & 461.9 & 533.9 & 149.2 & 108.5\\ 
$\data$02 &
PUMA0101301
 & 107.2 & 547.2 & 500.3 & 106.7 & 112.5\\ 
$\data$03 &
PUMA0800803
 & 107.2 & 446.1 & 571.7 & 120.3 & 107.2\\ 
$\data$04 &
PUMA1304600
 & 107.2 & 408.1 & 426.3 & 120.8 & 109.8\\ 
$\data$05 &
PUMA1703529
 & 107.2 & 435.3 & 426.3 & 134.9 & 110.9\\ 
$\data$06 &
PUMA1703531
 & 107.2 & 584.0 & 677.4 & 111.4 & 108.1\\ 
$\data$07 &
PUMA1901700
 & 107.2 & 395.1 & 443.6 & 119.1 & 110.4\\ 
$\data$08 &
PUMA2401004
 & 107.2 & 369.3 & 329.0 & 109.6 & 107.5\\ 
$\data$09 &
PUMA2602702
 & 107.2 & 467.8 & 472.0 & 146.0 & 109.2\\ 
$\data$10 &
PUMA2801100
 & 107.2 & 543.7 & 558.2 & 117.8 & 110.8\\ 
$\data$11 &
PUMA2901901
 & 107.2 & 485.2 & 464.4 & 126.5 & 110.8\\ 
$\data$12 & 
PUMA3200405
 & 107.2 & 329.1 & 301.0 & 122.9 & 108.4\\ 
$\data$13 &
PUMA3603710
 & 107.2 & 300.3 & 293.3 & 85.7 & 108.8\\ 
$\data$14 &
PUMA3604010
 & 107.2 & 399.9 & 386.5 & 129.8 & 111.3\\ 
$\data$15 &
PUMA5101301
 & 107.2 & 396.1 & 369.5 & 139.2 & 107.2\\ 
$\data$16 &
PUMA5151255
 & 107.2 & 330.7 & 280.3 & 139.1 & 107.8\\ 
\hline\end{tabular}
\caption{Squared Error for Sum Query (overall $\epsilon=0.5)$)}\label{table:sum}
\end{table}

In Table \ref{table:sum}, we show the squared error of these postprocessing methods for the sum query. The NNLS and MaxFitting baselines perform poorly for this query, with errors typically 4-5x those of the OLS method (which is close to the variance of the original noisy answer to the sum query). Meanwhile Sequential and ReWeighted fitting perform much better. Standard errors were roughly 2-6\% of the reported metrics (omitted for space, but shown in the supplementary materials).

\begin{table}
\begin{tabular}{|l|rr|rr|rr|rr|rr|}
\multicolumn{1}{c}{}
 & \multicolumn{2}{c}{\textbf{OLS}}
 & \multicolumn{2}{c}{\textbf{NNLS}}
 & \multicolumn{2}{c}{\textbf{MaxFit}}
 & \multicolumn{2}{c}{\textbf{Seq}}
 & \multicolumn{2}{c}{\textbf{ReWeight}}
\\ 
\multicolumn{1}{c|}{\textbf{Data}}& Total & Max & Total & Max & Total & Max & Total & Max & Total & Max \\\hline
$\data$01
 & 10516.5 & 124.0  & 344.2 & 147.4  & 443.7 & 173.1  & 437.3 & 283.6  & 159.2 & 78.4 \\ 
$\data$02
 & 23906.2 & 142.9  & 809.0 & 135.6  & 910.7 & 144.6  & 782.9 & 179.7  & 731.3 & 209.8 \\ 
$\data$03
 & 23906.2 & 142.9  & 1179.8 & 107.5  & 1235.3 & 125.4  & 1171.7 & 189.8  & 1123.8 & 141.9 \\ 
$\data$04
 & 23906.2 & 142.9  & 1313.0 & 111.9  & 1385.5 & 142.3  & 1049.1 & 126.3  & 1264.4 & 136.0 \\ 
$\data$05
 & 23906.2 & 142.9  & 1243.8 & 105.3  & 1257.2 & 96.7  & 1019.1 & 114.3  & 1285.7 & 160.5 \\ 
$\data$06
 & 23906.2 & 142.9  & 562.2 & 94.9  & 599.0 & 72.1  & 429.9 & 112.9  & 409.8 & 78.8 \\ 
$\data$07
 & 23906.2 & 142.9  & 1516.1 & 115.9  & 1665.9 & 129.7  & 1312.1 & 156.9  & 1617.1 & 205.0 \\ 
$\data$08
 & 23906.2 & 142.9  & 1954.4 & 130.0  & 1971.8 & 147.8  & 1983.4 & 311.3  & 1760.1 & 168.9 \\ 
$\data$09
 & 23906.2 & 142.9  & 977.2 & 100.0  & 956.4 & 109.4  & 843.4 & 121.7  & 930.1 & 156.2 \\ 
$\data$10
 & 23906.2 & 142.9  & 686.9 & 97.5  & 705.7 & 79.0  & 534.2 & 92.7  & 516.0 & 78.7 \\ 
$\data$11
 & 23906.2 & 142.9  & 944.4 & 100.4  & 919.2 & 103.2  & 809.4 & 131.6  & 888.2 & 138.2 \\ 
$\data$12
 & 23906.2 & 142.9  & 2189.2 & 119.6  & 2191.5 & 134.7  & 1918.5 & 142.3  & 2336.1 & 259.1 \\ 
$\data$13
 & 23906.2 & 142.9  & 2884.1 & 119.2  & 3088.6 & 149.1  & 2484.2 & 140.7  & 2870.4 & 166.1 \\ 
$\data$14
 & 23906.2 & 142.9  & 1432.5 & 105.9  & 1442.3 & 120.6  & 1262.1 & 122.7  & 1448.6 & 194.0 \\ 
$\data$15
 & 23906.2 & 142.9  & 1474.7 & 108.3  & 1498.6 & 101.8  & 1394.5 & 203.4  & 1392.9 & 153.2 \\ 
$\data$16
 & 23906.2 & 142.9  & 2239.7 & 130.3  & 2274.1 & 124.3  & 2079.0 & 178.5  & 2123.0 & 172.8 \\ 
\hline\end{tabular}
\caption{Squared Errors Id Query (overall $\epsilon=0.5$).}\label{table:id}
\end{table}

For Table \ref{table:id} we examine the expected errors of each cell query (i.e., $q_i$ is the number of people in cell $i$). We find the cell with the largest expected error and report it (the ``Max'' column). We also find the total squared error of the cell queries and report them in the ``Total'' column. Again, the standard errors are roughly 2-6\% of the reported metrics, except that they are sometimes higher for Max and Sequential fitting since averages were only computing on the subset of runs for which the optimizer did not fail.

Generally, NNLS performed slightly better in terms of the maximum expected error compared to ReWeight, although their total errors are comparable and ReWeight significantly outperforms NNLS on the sum query.

Overall, these experiments and our supplementary material show that both ReWeight and Sequential fitting (though not perfect) avoid incidents where there are extremely high errors (unlike NNLS and Max Fitting for sum queries), and this is important in practice. ReWeight and Sequential fitting have similar performance. ReWeight is faster while Sequential needs significant tuning of optimizers in order to succeed. However, one advantage of Sequential is its algorithmic transparency -- it can directly prioritize queries for the tradeoffs caused by the uncertainty principle (in our experiments, the sum query had highest priority for Sequential Fitting).

%% file: related.tex
The requirement to produce microdata is an example of consistency in privacy-preserving query answering.
A variety of work \cite{barakCDKMT07:holistic,consistency,hbtree,lioptim,DPSD,dpcube,LMHMR2015,JaewooKDD2015,mwem} has shown that creation of a privacy-preserving data synopsis from which all queries are answered can improve query accuracy under a variety of metrics such as maximum simultaneous error and total error. However, it is known that the production of privacy-preserving microdata comes at the expense of increased computational cost \cite{Vadhan2017,DNT,UV2011}. For example, under standard complexity assumptions \cite{UV2011}, there is no polynomial-time algorithm for generating privacy-protected synthetic data whose two-way marginals are all accurate.

Aside from the computational price, Balcer and Vadhan \cite{Balcer_Vadhan_2019} also recently showed a statistical price of privacy-protected synthetic data. They considered releasing different kinds of privacy-protected representations of nonnegative noisy histograms (for example, releasing the top-k noisy cells under $\epsilon$-DP had a $\log^2(d/k)$ penalty term for squared error), but assumed the value of the sum query was publicly known in their work. Our constructions are based on their proof techniques (see discussion after Theorem \ref{thm:lowerbound}). 


%% file: conc.tex
Public-use data have many different end-users, so a single aggregated performance measure, such as total error across all queries, is not a reliable measure of data quality. The accuracy of each query is important, which implies multiple conflicting quality criteria for public-use data. Thus an important direction for future work is to identify all tradeoffs in privacy-preserving microdata as well as algorithms with provable guarantees on instance-optimality (i.e., improve performance on datasets that do not trigger the uncertainty principles). 

\subsection*{Broader Impact}
The uncertainty principle presented in this paper (as well as the cost of microdata results in \cite{Balcer_Vadhan_2019}) along with the known computational price of generating microdata suggests that organizations should also consider alternative formats for their privacy-protected data products. The uncertainty principle can be avoided by releasing noisy query answers that are allowed to be negative or by producing weighted datasets that can feature negative weights (however, adding a sparsity requirement could re-introduce systematic errors \cite{Balcer_Vadhan_2019}). 
Such alternative formats may also require educating and providing training materials to end-users. If an organization nevertheless decides to produce privacy-protected microdata, then microdata-generating algorithms should be designed as postprocessing algorithms that convert unbiased noisy measurements into microdata (so that the ``statistics-friendly'' noisy measurements can also be released and studied by data scientists). Further research into such postprocessing algorithms is needed to mitigate the effects of the uncertainty principle.

%% file: ack.tex
\begin{ack}
We thank Salil Vadhan for helpful discussions that allowed us to sharpen the lower bound results. Affiliations are provided solely for the purpose of identification. All work was performed under the supervision of the U.S. Census Bureau as part of the authors' employment or contractual work product. The views and opinions in this article are those of the authors and do not represent the policy or official position of the U.S. Government, the U.S. Department of Commerce, the U.S. Census Bureau, the U.S. Department of Homeland Security, Knexus, or Tumult Labs. 

\vspace{0.1cm}
\noindent\textbf{Competing interests:} None.

\vspace{0.1cm}
\noindent\textbf{Additional revenues related to this work:} None.

\end{ack}

%% file: appendix.tex
\appendix

\section{Appendix (Supplementary Material)}

\subsection{Proofs Lower Bound Results}
\printProofs[lowerbound]

\newpage

\subsection{Proof of Upper Bound Results}
\input{usefullemmas}

\printProofs[upperbound]

\input{datasupp}

\clearpage
\input{allexp/allexperiments.tex}


%% file: usefullemmas.tex
We first need some facts about Gaussian and Laplace random variables.
\begin{lemma}\label{lem:distributions}
Let $z_1,\dots, z_d$ be i.i.d. random variables from a distribution $F$.
\begin{itemize}
\item If $F$ is $N(0, \sigma^2)$ then
\begin{itemize}
\item $E[z_i^2]=\sigma^2$ for all $i$
\item $E[|z_i|]\leq \sigma$ for all $i$
\item $E[\max_i |z_i|]\in O(\sigma\sqrt{\log(d)}) $
\item $E[\max_i z_i^2] \in O(\sigma^2\log(d)) $
\end{itemize}
\item If $F$ is $Lap(1/\epsilon)$ then
\begin{itemize}
\item $E[z_i^2]=2/\epsilon^2$ for all $i$
\item $E[|z_i|]=1/\epsilon$ for all $i$
\item $E[\max_i |z_i|] \leq\frac{1}{\epsilon}(\ln(d) + 1)$
\item $E[\max_i z_i^2] \leq \frac{1}{\epsilon^2}(\ln^2(d) + 2\ln(d) + 2)$
\end{itemize}
\end{itemize}
\end{lemma}
\begin{proof}
The variance of a Gaussian is known to be $\sigma^2$ and that of the Laplace distribution is known to be $2/\epsilon^2$. 

The absolute value of a Laplace is an Exponential random variable with rate $\epsilon$ and so the expectation is $1/\epsilon$.
Next, by Jensen's inequality $\left(E[|z_i|]\right)^2\leq E[z_i^2]$ and so $E[|z_i|] \leq \sqrt(E[z_i^2])$. Thus, in the case of a Gaussian, this is upper bounded by $\sigma$.

\textbf{To compute the expectation of the maxes}, we note that if $z^\prime_i$ follows the Lap$(1)$ distribution, then $z^\prime/\epsilon$ follows the Lap$(1/\epsilon)$ and if $z^\prime$ follows $N(0,1)$ then $\sigma z$ follows the $N(0, \sigma^2)$ distribution. Thus we compute the expectations under the assumption that the scale variables are $1$ and then we multiply by $1/\epsilon$ or $\sigma$ for the first moment, and $1/\epsilon^2$ or $\sigma^2$ for the second moment to get the results for $z_i$ from the results for $z^\prime_i$. 

Next we let  $G$ be  the cdf of a continuous nonnegative random variable and $g$ the corresponding pdf. Then for any $p\geq 1$,  
\begin{align*}
E_{X\sim G}\left[X^p\right] &= \int_{0}^\infty x^p g(x)~dx = \int_{0}^\infty g(x)\left(\int_0^\infty p t^{p-1}1_{\{t\leq x\}} ~dt\right)~dx\\
&= \int_{0}^\infty p t^{p-1}\left(\int_0^\infty g(x) 1_{\{t\leq x\}} ~dx\right)~dt = \int_0^\infty pt^{p-1} (1-G(t))~dt
\end{align*}

Now we let $F_+$ be the cdf of $|z^\prime_1|$ (the random variables with location parameter 1), and let $G$ be the distribution of $\max_i |z^\prime_i|$. Then for all $t$, $G(t)=F_+(t)^d$ and also by the union bound, $1-F_+(t)^d = 1-G(t)\leq d(1-F_+(t))$. For any $\gamma > 0$,
\begin{align*}
E\left[\max_i |z^\prime_i|^p\right]&=\int_0^\infty pt^{p-1} (1-F_+(t)^d)~dt\\
&= \int_0^\gamma pt^{p-1} (1-F_+(t)^d)~dt   + \int_\gamma^\infty pt^{p-1} (1-F_+(t)^d)~dt\\
&\leq \int_0^\gamma pt^{p-1}~dt + \int_\gamma^\infty pt^{p-1} d(1-F_+(t))~dt\\
&= \gamma^p + \int_\gamma^\infty pt^{p-1} d(1-F_+(t))~dt
\end{align*}

\textbf{For the Laplace distribution}, $F_+(t)=1-e^{-t}$ thus, for any $\gamma > 0$
\begin{align*}
E\left[\max_i |z^\prime_i|\right] \leq \gamma + \int_\gamma^\infty  de^{-t}~dt = \gamma + de^{-\gamma}\\
E\left[\max_i |z^\prime_i|^2\right] \leq \gamma^2 + \int_\gamma^\infty  2tde^{-t}~dt = \gamma^2 + 2\gamma de^{-\gamma} + 2de^{-\gamma}\\
\end{align*}
Setting $\gamma=\ln(d)$ and converting from $z^\prime_i$ to $z_i$, we get $E[\max_i |z_i|]\leq\frac{1}{\epsilon}(\ln(d) + 1)$ and $E[\max_i |z_i|^2]\leq \frac{1}{\epsilon^2}(\ln^2(d) + 2\ln(d) + 2)$.

\textbf{For the Gaussian distribution}, a well-known tail bound on the Gaussian is that $1-F_+(t)\leq \frac{2}{t}\frac{1}{\sqrt{2\pi}}e^{-t^2/2}$. Thus we get

\begin{align*}
E\left[\max_i |z^\prime_i|\right] &\leq \gamma + \int_\gamma^\infty  d\frac{2}{t}\frac{1}{\sqrt{2\pi}}e^{-t^2/2}~dt \leq \gamma + \frac{2d}{\gamma}\int_\gamma^\infty \frac{1}{\sqrt{2\pi}}e^{-t^2/2}~dt\\
&\leq \gamma + 2d\frac{2}{\gamma^2}\frac{1}{\sqrt{2\pi}}e^{-\gamma^2/2}\\
E\left[\max_i |z^\prime_i|^2\right] &\leq \gamma^2 + \int_\gamma^\infty  2td\frac{2}{t}\frac{1}{\sqrt{2\pi}}e^{-t^2/2}~dt =\gamma^2 + 4d\int_\gamma^\infty \frac{1}{\sqrt{2\pi}}e^{-t^2/2}~dt\\
&= \gamma^2 + 4d\frac{1}{\gamma}\frac{1}{\sqrt{2\pi}}e^{-\gamma^2/2}
\end{align*}
Setting $\gamma=\sqrt{2\ln(d)}$ and converting from $z^\prime_i$ to $z_i$, we get $E[\max_i |z_i|]\leq \sigma(\sqrt{2\ln(d)} + \frac{4}{\sqrt{2\pi}}\frac{1}{2\ln(d)})$ and $E[\max_i |z_i|^2]\leq \sigma^2(2\ln(d) + \frac{4}{\sqrt{2\pi}}\frac{1}{\sqrt{2\ln(d)}})$.
\end{proof}

We next need a technical lemma about the solution to a constrained nonnegative least squares problem.

\begin{lemma}\label{lem:cnnls}
Let $a_1,\dots, a_d$ be real numbers and let $a_*\geq 0$. The solution to the optimization problem
\begin{align*}
\arg\min_{x_1,\dots, x_d} &\frac{1}{2}\sum_{i=1}^d (x_i-a_i)^2\\
 \text{s.t. } &\sum_{i=1}^d x_i = a_*\\
 \phantom{\text{s.t. }} &x_i\geq 0, \text{ for $i=1,\dots, d$}
\end{align*}
is $x_i=\max\{a_i-\gamma, 0\}$ (for all $i$) where $\gamma$ is chosen so that $\sum_{i=1}^d \max\{0, a_i-\gamma\}=a_*$.
\end{lemma}
\begin{proof}
Let us use the shorthand $(a-\gamma)_+$ to mean $\max\{0, a-\gamma\}$.

First, it is easy to see that by continuity, there exists a $\gamma$ such that $\sum_i (a_i-\gamma)_+=a_*$. 

The gradient of the objective function with respect to the $x_i$ is:
\begin{align*}
\frac{\partial \text{obj}}{\partial x_i} = (x_i-a_i) = (a_i-\gamma)_+-a_i
\end{align*}
and if this choice of $x_i$ is not optimal, then any descent direction $(y_1,\dots, y_n)$ (i.e., for which $x_1+y_1, \dots, x_1+y_2$ is feasible and reduces the objective function) must satisfy (1) $\sum_{i=1}^d y_i = 0$ to maintain feasibility of the equality constraint, (2) $\sum_{i=1}^d y_i ((a_i-\gamma)_+ - a_i) < 0$ to be a descent direction, (3) $y_i\geq 0$ when $a_i \leq \gamma$ and $y_i \geq \gamma-a_i$ when $a_i>\gamma$ to maintain nonnegativity of $x_i+y_i\equiv (a_i-\gamma)_+ + y_i$.

Now,
\begin{align*}
\sum_{i=1}^d y_i((a_i-\gamma)_+ - a_i) &\sum_{i: a_i > \gamma} y_i((a_i-\gamma)_+ - a_i) + \sum_{i: a_i \leq \gamma} y_i((a_i-\gamma)_+ - a_i)\\
&=-\gamma \sum_{i: a_i > \gamma} y_i - \sum_{i: a_i \leq \gamma} y_i a_i\\
&=-\gamma \sum_{i: a_i > \gamma} y_i-\sum_{i: a_i \leq \gamma} y_i\gamma\\
&\quad\text{since feasibility of $x_i+y_i$ requires $y_i\geq 0$ when $a_i\leq \gamma$}\\
&=-\gamma \sum_{i=1}^d y_i =0 \quad\text{since feasibility requires $\sum_i y_i=0$}
\end{align*}
contradicting that $y_1,\dots, y_d$ is a descent direction.

\end{proof}

%% file: datasupp.tex
\newgeometry{margin=1in}
\section{Full Data Benchmark Description}

Our benchmarks contain 15 real datasets and 16 synthetic datasets. The datasets are designed to be small enough to enable thousands of runs (in order to compute expected squared errors) but large enough to clearly illustrate postprocessing errors and present a challenge to many open-source optimizers.

\subsection{Real Datasets}

The real datasets are drawn from the 2016 American Community Survey Public Use Microdata Sample (PUMS) \cite{2016pums}, which provides records for geographies known as Public Use Microdata Areas (PUMA).

To create a benchmark data set that adequately captured the diversity of real world demographic data, we drew from outlier geographies in the 2016 ACS PUMS. We chose 15 Public Use Microdata Areas whose data distributions had been identified as conflicting significantly with the majority distributions in their states, according to the k-marginal metric used by NIST in their Differential Privacy challenge \cite{NISTDesign}.  The data spanned historically redlined areas, a variety of immigrant communities, wealthy and diverse urban neighborhoods, rural agricultural communities, and included every major region in the United States. 

For each of the 15 regions, we created a $9\times 24$ Race by Hispanic Origin histogram. These were two separate questions in the ACS questionnaire. Although the questionnaire allowed respondents to select multiple races (from a list of 15 categories and 3 fill-in text boxes), most individuals belong to three or fewer races, and the 2016 ACS PUMS did not include detailed racial breakdowns for individuals with more than 3 races.  To mimic the extreme sparsity and geographically diverse correlation patterns in the multi-racial checkbox variable, we selected two variables (called RAC1P and HISP; full definitions below): a smaller race variable with 9 possible values which primarily records single races, and a detailed Hispanic origin variable with 24 possible values.  Any of the 216 possible combinations of race and Hispanic origin is valid; individuals of all races have origins from all across Latin America.   However, in any given community the vast majority of these counts will be zero, resulting in sparse distributions.  At the same time, communities with different immigration histories will differ significantly with respect to which counts are nonzero and in the size of the other counts.   Algorithms which performed well across all cases in the PUMS benchmark data set should be expected perform well on the edge case complexities of national data.  

\begin{verbatim}
RAC1P   
    Recoded detailed race code
1. White alone
2. Black or African American alone
3. American Indian alone
4. Alaska Native alone
5. American Indian and Alaska Native tribes specified; or American
 Indian or Alaska Native, not specified and no other races
6. Asian alone
7. Native Hawaiian and Other Pacific Islander alone
8. Some Other Race alone
9. Two or More Races

HISP 
Detailed Hispanic origin
01. Not Spanish/Hispanic/Latino
02. Mexican
03. Puerto Rican
04. Cuban
05. Dominican
06. Costa Rican
07. Guatemalan
08. Honduran
09. Nicaraguan
10. Panamanian
11. Salvadoran
12. Other Central American 
13. Argentinean
14. Bolivian
15. Chilean
16. Colombian
17. Ecuadorian
18. Paraguayan
19. Peruvian
20. Uruguayan
21. Venezuelan
22. Other South American 
23. Spaniard
24. All Other Spanish/Hispanic/Latino
\end{verbatim}

\subsection{Synthetic Data}
The synthetic data are modeled after the proofs of our lower bound results. The main idea is that suppose  noise from a distribution $F$ is added to a histogram, and that there are $k$ zero cells and one cell with a count of $C$ in that histogram.
Based on the noisy cell values, it is difficult to guess which cell had value $C$ when $C$ is smaller than the median of the distribution of $\max\{X_1,\dots, X_k\}$ (whose CDF is $F^k(t)$), where each $X_i\sim F$. Thus we created datasets with sparsity patterns.

Each histogram had 100 elements, from which we created a 1-dimensional version (a 100-element vector) and a 2-dimensional version (reshaping it to a $10\times 10$ histogram). In all of the datasets, the first histogram cell is relatively large (10,000) and should be easy to distinguish from 0 based on the noisy counts (although ordinary nonnegative least squares fails to do so).

The synthetic histograms come from 4 categories, defined as follows:

\begin{itemize}
\item \textbf{Level}. In the \underline{Level \textbf{k}} histograms, all cells have the same value $k$ (except the first, which has value 10,000). The benchmarks include 1- and 2-dimensional versions of \underline{Level0} (i.e., only the first element is nonzero), \underline{Level1}, \underline{Level16}, and \underline{Level32}. The Level1 dataset presents a tricky case where each cell (other than the first), based on its noisy value, may look similar to 0, but the overall sum of these small cells is clearly distinguishable from 0. The Level16 and Level32 datasets are designed to force algorithms to try to estimate the number of cells that are likely to have true value of 0. Note that 16 is roughly the 40th percentile of the distribution of $\max{X_1,\dots, X_{100}}$ when each $X_i$ has the Laplace$(1/\epsilon)$ distribution with $\epsilon=0.25$. So having a few cell noisy cell counts near 16 is possible when a histogram is mostly 0, but having many noisy counts near 16 is a sign that the histogram is not sparse.
\item \textbf{Stair}. The \underline{Stair} data is a histogram that looks like this: $[10000, 1, 2, 3, 4, \dots]$ in one dimension (and is reshaped into a $10\times 10$ matrix in 2 dimensions. It is designed to simulate a dataset with small, medium, and large values.
\item\textbf{Step}. The \underline{Step k} dataset is a step function. The first element is 10000, the next 49 are 0 and the last 50 are $k$. This is an interpolation between the sparse dataset synthetic dataset Level0 and  Level $k$. For our benchmark, we use \underline{Step16} (i.e., $k=16$) as a dataset of medium difficulty and \underline{Step50} as an easy dataset.
\item\textbf{SplitStairs}. The \underline{SplitStairs} dataset is an interpolation between Stair and a very sparse dataset. The first half looks like the Stair dataset but cells 50 until the end all have value 0. This ensures that all true cell counts that can be dominated by 50 random zero-mean Laplace random variables are represented in the dataset. 
\end{itemize}

Combined, these synthetic datasets give 8 1-dimensional histograms (4 Level, 1 Stair, 2 Step, 1 SplitStairs) and 8 2-dimensional histograms.

%% file: allexp/allexperiments.tex
\begin{center}
\textbf{\Large{Complete experimental results.}}
\end{center}

Here, we present our full experimental results. The datasets used are the PUMS datasets (2-dimensional), the 1-dimensional synthetic data, and the 2-dimensional synthetic data. These datasets are described in the appendix of the full version of the paper, which appears in the supplementary material file.

For the one-dimensional datasets, we use either the Laplace mechanism (for pure differential privacy) or the Gaussian mechanism (for zCDP) to obtain noisy answers to:
\begin{itemize}
\item The sum query (the sum of the histogram cells)
\item The  identity queries (the count in each cell).
\end{itemize}

For the two-dimensional datasets, we use either the Laplace or Gaussian mechanisms to obtain noisy answers to:
\begin{itemize}
\item The sum query (the sum of the histogram cells)
\item The  identity queries (the count in each cell).
\item The marginal on the first dimension.
\item The marginal on the second dimension.
\end{itemize}

We use the NNLS (referred to as nnlsalg in the tables), Max fitting, Sequential Fitting, and Weighted Fitting (with confidence parameter $0.99$) postprocessing methods to obtain the privacy preserving positively weighted data $\widetilde{D}$. Sequential Fitting prioritizes queries in the order listed above.
We also use OLS fitting (NNLS fitting without the nonnegativity constraints), which is referred to as olsalg in the experiments. The OLS fitting method is known to improve the squared error of the queries compared to the original noisy answers (this is a consequence of the Gauss-Markov theorem) but does not result in a positively weighted dataset. Hence the goal of the methods is not to do much worse than the OLS fitting method.

The code was written in Julia. In order to make the code fully open source, we experimented with several open source solvers compatible with Julia's JuMP framework. Out of these, the COSMO solver performed the best. However, the relatively complex multi-stage optimizations in  Max fitting and Sequential fitting caused problems. In some cases the solver claimed infeasibility for problems in latter stages of the optimization (likely due to poor quality solutions in earlier stages), numerical errors, or slow convergence (hitting the iteration limit). To reduce the chance of poor solutions in earlier stages of an optimization, we set the absolute and relative tolerances to 1e-7 and an iteration limit of 20,000, which is 4 times the default. We also converted equality constraints of the form $x=constant$ to $x\leq constant + 0.001$ and $x \geq constant - 0.001$. For the Max Fitting solve, after it gets an $L_\infty$ distance estimate in the first stage of the solve, we added a slack of $0.01$ to this distance to prevent it from failing in the second stage.

Despite tuning parameter and setting slack tolerances to equalities and inequalities, not all runs were successful, so we only kept the ones where all stages of the optimization were optimal. This likely optimistically biased the results of Max fitting and Sequential fitting and increased their estimated standard errors.

These optimization problems did not affect OLS, NNLS, or the Weighted Fitting approaches.

Each experiment is an average over 1000 runs (thus the expected error of a query is estimated the average of its errors across 1000 runs). However, for more complex constrained methods, the average was among fewer runs if some stage of the multi-stage optimization failed to find an optimal solution.

In each table, we evaluate the error of different queries. 
\begin{itemize}
\item For the Sum query (as in Table \ref{table:experiment1:Sum:sq}), we display its expected error  along with estimated standard deviation.
\item For the Identity queries (as in Table \ref{table:experiment1:Id:sq}), each cell $i$ in the histogram corresponds to a query $q_i$ (the count in that cell).  For each cell $i$, we estimate its expected squared error $e_i=E[((q_i\data) - q_i(\widetilde{\data}))^2]$ by averaging the error across trials. Then we report $\max_i e_i$ and $\sum_i e_i$ along with standard errors. Again, we emphasize that our Max metric is $\max_i E[((q_i\data) - q_i(\widetilde{\data}))^2]$ and not outlier error $E[\max_i(((q_i\data) - q_i(\widetilde{\data})))^2]$.
\item For the two dimensional datasets, we also have tables for each marginal and report the max and total squared errors as for the identity queries.
\end{itemize}

Note that the goal is to avoid extreme errors that are much larger than the OLS error.

The experiments are organized first by privacy definition (pure DP and zCDP). Within each privacy definition, we first present results for the $1$-dimensional synthetic data (for 3 privacy parameters) followed by the $2$-dimensional synthetic data (for 3 privacy parameters) followed by the PUMS data (for 3 privacy parameters).

\section{Pure Differential Privacy}

\input{allexp/experiment1-Id-basic.tex}
\input{allexp/experiment1-Sum.tex}
\input{allexp/experiment2-Id-basic.tex}
\input{allexp/experiment2-Sum.tex}
\input{allexp/experiment3-Id-basic.tex}
\input{allexp/experiment3-Sum.tex}
\input{allexp/experiment4-Marg1-basic.tex}
\input{allexp/experiment4-Marg2-basic.tex}
\input{allexp/experiment4-Id-basic.tex}
\input{allexp/experiment4-Sum.tex}
\input{allexp/experiment5-Marg1-basic.tex}
\input{allexp/experiment5-Marg2-basic.tex}
\input{allexp/experiment5-Id-basic.tex}
\input{allexp/experiment5-Sum.tex}
\input{allexp/experiment6-Marg1-basic.tex}
\input{allexp/experiment6-Marg2-basic.tex}
\input{allexp/experiment6-Id-basic.tex}
\input{allexp/experiment6-Sum.tex}
\input{allexp/experiment7-Marg1-basic.tex}
\input{allexp/experiment7-Marg2-basic.tex}
\input{allexp/experiment7-Id-basic.tex}
\input{allexp/experiment7-Sum.tex}
\input{allexp/experiment8-Marg1-basic.tex}
\input{allexp/experiment8-Marg2-basic.tex}
\input{allexp/experiment8-Id-basic.tex}
\input{allexp/experiment8-Sum.tex}
\input{allexp/experiment9-Marg1-basic.tex}
\input{allexp/experiment9-Marg2-basic.tex}
\input{allexp/experiment9-Id-basic.tex}
\input{allexp/experiment9-Sum.tex}

\clearpage
\section{zCDP Experiments}

\input{allexp/experiment10-Id-basic.tex}
\input{allexp/experiment10-Sum.tex}
\input{allexp/experiment11-Id-basic.tex}
\input{allexp/experiment11-Sum.tex}
\input{allexp/experiment12-Id-basic.tex}
\input{allexp/experiment12-Sum.tex}
\input{allexp/experiment13-Marg1-basic.tex}
\input{allexp/experiment13-Marg2-basic.tex}
\input{allexp/experiment13-Id-basic.tex}
\input{allexp/experiment13-Sum.tex}
\input{allexp/experiment14-Marg1-basic.tex}
\input{allexp/experiment14-Marg2-basic.tex}
\input{allexp/experiment14-Id-basic.tex}
\input{allexp/experiment14-Sum.tex}
\input{allexp/experiment15-Marg1-basic.tex}
\input{allexp/experiment15-Marg2-basic.tex}
\input{allexp/experiment15-Id-basic.tex}
\input{allexp/experiment15-Sum.tex}
\input{allexp/experiment16-Marg1-basic.tex}
\input{allexp/experiment16-Marg2-basic.tex}
\input{allexp/experiment16-Id-basic.tex}
\input{allexp/experiment16-Sum.tex}
\input{allexp/experiment17-Marg1-basic.tex}
\input{allexp/experiment17-Marg2-basic.tex}
\input{allexp/experiment17-Id-basic.tex}
\input{allexp/experiment17-Sum.tex}
\input{allexp/experiment18-Marg1-basic.tex}
\input{allexp/experiment18-Marg2-basic.tex}
\input{allexp/experiment18-Id-basic.tex}
\input{allexp/experiment18-Sum.tex}